\documentclass{lmcs}

\pdfoutput=1

\usepackage{lastpage}
\renewcommand{\dOi}{14(3:25)2018}
\lmcsheading{}{1--\pageref{LastPage}}{}{}%
{Feb.~13,~2018}{Sep.~28,~2018}{}

\usepackage{proof}
\usepackage{amsthm}	
 
\newlength{\myheight}

\usepackage{amssymb,graphicx,color}
\usepackage[all]{xy}

\usepackage{pgf}
\usepackage{tikz}
\usetikzlibrary{arrows,shapes,snakes,automata,backgrounds,petri,fit,positioning}
\tikzstyle{node}=[circle, draw=black, minimum size=1mm]
\tikzstyle{trans}=[font=\scriptsize]
\tikzstyle{lab}=[font=\small]

\newcommand{\drawplace}{*[o]=<.9pc,.9pc>{\ }\drop\cir{}}
\newcommand{\drawmarkedplace}{*[o]=<.9pc,.9pc>{\bullet}\drop\cir{}}

\newcommand{\drawpersistentplace}{*[o]=<.9pc,.9pc>[Foo]{\ }}

\newcommand{\nameplacedown}[1]{\POS[]-<0pc,.8pc>\drop{\scriptstyle{#1}}}
\newcommand{\nameplaceright}[1]{\POS[]+<.8pc,0pc>\drop{\scriptstyle{#1}}}
\newcommand{\nameplaceleft}[1]{\POS[]-<.8pc,0pc>\drop{\scriptstyle{#1}}}

\newcommand{\drawtrans}[1]{*=<1.5pc,.9pc>{#1}\drop\frm{-}}

\newcommand{\cat}[1]{\ensuremath{\mathsf{#1}}}

\newcommand{\fun}[1]{\ensuremath{\mathcal{#1}}}

\newcommand{\mon}[1]{\ensuremath{{#1}^\oplus}}

\newcommand{\evstr}{\textsc{es}}
\newcommand{\pow}[1]{\ensuremath{\mathbf{2}^{#1}}}
\newcommand{\powfin}[1]{\ensuremath{\mathbf{2}_f^{#1}}}

\newcommand{\set}[1]{\ensuremath{[\!\![{#1}]\!\!]}}

\newcommand{\pre}[1]{\ensuremath{{}^\bullet{#1}}}
\newcommand{\post}[1]{\ensuremath{{#1} {^\bullet}}}
\newcommand{\pn}{\ensuremath{\cat{PN}}}
\newcommand{\occ}{\ensuremath{\cat{ON}}}
\newcommand{\occeq}{\ensuremath{\cat{OE}}}

\newcommand{\es}{\ensuremath{\cat{ES}}}
\newcommand{\ces}{\ensuremath{\cat{cES}}}

\newcommand{\les}{\ensuremath{\cat{{\ell}ES}}}

\newcommand{\WDom}{\ensuremath{\cat{wDom}}}
\newcommand{\occquot}[1]{\ensuremath{\fun{Q}({#1})}}

\newcommand{\preUnf}[1]{\ensuremath{\fun{U}^{pre}({#1})}}
\newcommand{\unf}[1]{\ensuremath{\fun{U}({#1})}}

\newcommand{\Ioc}[1]{\ensuremath{\fun{I}({#1})}}

\newcommand{\occEs}[1]{\ensuremath{\fun{E}({#1})}}
\newcommand{\esOcc}[1]{\ensuremath{\fun{O}({#1})}}

\newcommand{\nat}{\ensuremath{\mathbb{N}}}

\newcommand{\eqclass}[2][]{\ensuremath{[{#2}]_{\scriptscriptstyle {#1}}}}

\newcommand{\quotient}[2]{\ensuremath{{#1}_{\scriptscriptstyle {#2}}}}

\newcommand{\zunf}[0]{\ensuremath{\zunf}}

\newcommand{\conf}[1]{\ensuremath{\mathit{Conf}({#1})}}
\newcommand{\hist}[1]{\ensuremath{\mathbf{H}({#1})}}

\newcommand{\conn}[1]{\ensuremath{\stackrel{#1}{\frown}}}

\newcommand{\interval}[2][1]{\ensuremath{[{#1},{#2}]}}

\newcommand{\fire}[1]{\ensuremath{[{#1}\rangle}}

\newcommand{\conc}[1]{\ensuremath{\mathit{co}({#1})}}
\newcommand{\sconc}[1]{\ensuremath{\mathit{sco}({#1})}}

\newcommand{\Mark}[1]{\ensuremath{\mathit{m}({#1})}}

\newcommand{\struct}{\ensuremath{\rightsquigarrow}}

\newcommand{\depth}[1]{\ensuremath{\mathit{d}({#1})}}
\newcommand{\mxdepth}[1]{\ensuremath{\mathit{md}({#1})}}
\newcommand{\cardmaxdepth}[1]{\ensuremath{\mathit{\#md}({#1})}}

\newcommand{\pl}[1]{\ensuremath{#1}_{\mathrm{s}}}
\newcommand{\tr}[1]{\ensuremath{#1}_{\mathrm{t}}}

\newcommand{\pe}[1]{\ensuremath{#1}^{\mathrm{p}}}
\newcommand{\npe}[1]{\ensuremath{#1}^{\mathrm{n}}}

\begin{document}

\title{Event Structures for Petri nets with Persistence}

%%% LMCS proof changes
\author[P.~Baldan]{Paolo Baldan\rsuper{a}}	\address{\lsuper{a}University of Padova, Italy}	\email{baldan@math.unipd.it}  
\author[R.~Bruni]{Roberto Bruni\rsuper{b}}	\address{\lsuper{b}University of Pisa, Italy}
\author[A.~Corradini]{Andrea Corradini\rsuper{b}}	\address{\vskip-7pt}	
\author[F.~Gadducci]{Fabio Gadducci\rsuper{b}}	\address{\vskip-7pt}
\email{bruni@di.unipi.it}  
\email{andrea@di.unipi.it}  
\email{gadducci@di.unipi.it}  
\email{ugo@di.unipi.it}  
\author[H.~Melgratti]{Hernan Melgratti\rsuper{c}}
\address{\lsuper{c}University of Buenos Aires - Conicet, Argentina}	\email{hmelgra@dc.uba.ar}  
\author[U.~Montanari]{Ugo Montanari\rsuper{b}}
%%%%%%%%%%

\begin{abstract}
  Event structures are a well-accepted model of concurrency. In a
  seminal paper by Nielsen, Plotkin and Winskel, they are used to
  establish a bridge between the theory of domains and the approach to
  concurrency proposed by Petri. A basic role is played by an
  unfolding construction that maps (safe) Petri nets into a subclass
  of event structures, called prime event structures, where each event
  has a uniquely determined set of causes. Prime event structures, in
  turn, can be identified with their domain of configurations. At a
  categorical level, this is nicely formalised by Winskel as a chain
  of coreflections.

  Contrary to prime event structures, general event structures allow
  for the presence of disjunctive causes, i.e., events can be enabled
  by distinct minimal sets of events.  In this paper, we extend the
  connection between Petri nets and event structures in order to
  include disjunctive causes.  In particular, we show that, at the
  level of nets, disjunctive causes are well accounted for by
  persistent places.  These are places where tokens, once generated,
  can be used several times without being consumed and where multiple
  tokens are interpreted collectively, i.e., their histories are
  inessential. Generalising the work on ordinary nets, Petri nets with
  persistence are related to a new subclass of general event
  structures, called \emph{locally connected}, by means of a chain of
  coreflections relying on an unfolding construction.
\end{abstract}

\maketitle

\section*{Prelude}

Among the multitude of his research interests, Furio has been working
on the foundations of concurrency with special attention to the
mathematical domains required for defining the semantics of concurrent
systems. We like to recall his work on hyperuniverses as models for
processes and on the final semantics of the pi-calculus achieved by
means of an higher order presentation via Logical Framework. It was in
the context of the European Project MASK (``Mathematical Structures
for Concurrency''), led by Jaco De Bakker, that most of us had the
first chance of working closely together with Furio.  Surely that was
one of the most relevant results of the project, paving the way for a
fruitful scientific collaboration that continued in the following years in the
framework of several research projects.
This paper is heartfully dedicated to him in the occasion of his 60th birthday.

\section{Introduction}

Petri nets have been introduced in the Ph.D. Thesis of Carl Adam Petri~\cite{Petri62} and soon have become one of the best known models of
concurrency~\cite{DBLP:books/daglib/0032298,DBLP:books/daglib/0023756,DBLP:series/eatcs/Gorrieri17}.
The conceptual simplicity of the model (multiset rewriting) and its
intuitive graphical presentation have attracted the interest of both
theoreticians and practitioners.  Nowadays Petri nets are widely
adopted across Computer Science and other disciplines such as Physics,
Chemistry, and
Biology~\cite{DBLP:conf/birthday/Abramsky08,DBLP:journals/it/KochRS14,DBLP:journals/sosym/Koch15}.
They provide a basic model that, on the one hand, offers
the ideal playground to study basic concepts of concurrent and
distributed
computations~\cite{murata89,DBLP:books/daglib/0007558,DBLP:books/daglib/0032051}
and, on the other hand, can be readily extended to experiment with
advanced features like structured data handling, read and inhibitor
arcs, mobility, reflection, time and stochastic
behavior~\cite{DBLP:conf/ac/Jensen86a,DBLP:conf/apn/ChristensenH93,DBLP:journals/acta/MontanariR95,DBLP:journals/mscs/AspertiB09,DBLP:conf/apn/Valk98,merlin74,ramchandani74,DBLP:conf/performance/DuganTGN84,DBLP:journals/tocs/MarsanCB84,Molloy:1985:DTS:4101.4110,DBLP:conf/apn/EmzivatDLR16}.

In this paper we are interested in the seminal work of
Winskel~\cite{Win:ES} on net unfolding, which has established a
tight connection between Petri nets and (prime algebraic) domains.
There it is shown that a chain of coreflections links the category of
safe nets to the category of prime event structures, which in turn is equivalent
to the category of prime algebraic domains. This is particularly
satisfactory since a coreflection essentially establishes that a
sub-category of abstract models can be found in a category of concrete
models, such that each concrete model can be assigned the best possible
abstract model.
The first step of the chain is an unfolding construction
that maps each net to a special kind of acyclic net (called
non-deterministic occurrence net) representing all
behaviours of the original net. From this an event structure can
be easily defined, by forgetting the places of the net.
Later these results have been extended to the more general class of
semi-weighted
nets~\cite{DBLP:conf/concur/MeseguerMS92,DBLP:journals/mscs/MeseguerMS97}.

Petri nets semantics is based on consuming and producing data
(i.e., tokens) from repositories (i.e., places). Operationally, 
reading a piece of information can be modelled by a transition that
consumes a token from a place and produces it again on the same place. 
However, from the point of view of
concurrency and causality such an encoding is not faithful as 
it disallows concurrent readings. Moreover, in many situations one is
interested in representing persistent information that once created
can be read but not consumed and such that its multiplicity (the
number of available instances) is not important. This is the case for instance of
classical 
logical conditions that once established to hold
can be used repeatedly for proving other conditions. Another example
is that of subversioning systems or cloud storage, where data changes
are logged and previous versions of stored files remain accessible
after an update.
Persistent information is also used in~\cite{DBLP:conf/ccs/CrazzolaraW01} 
to model security protocols, where the pool of messages exchanged by
participants over the network forms a knowledge base that remains
available for inspection and processing by attackers: in other words,
sent messages cannot be removed. 
In~\cite{DBLP:journals/corr/abs-1710-04570} persistent tokens are used
to remove confusion from acyclic nets and equip choices with standard
probabilistic distributions.
The equivalent framework of CPR nets~\cite{DBLP:journals/fuin/BonchiBCG09} 
has been proposed for modelling web services: persistency is needed to capture 
service availability, and the formalism has been used for describing protocols 
that are specified with the ontology-based OWL-S language.

Read arcs have been introduced in the literature to handle multiple
concurrent readings~\cite{DBLP:conf/apn/ChristensenH93} and their
concurrent semantics has been widely
investigated~\cite{VSY:UFP,BCM:CNAED}.
From the point of view of
causality, they are not expressive enough to model another interesting
phenomenon of persistent information that is the absence of
multiplicity: if the same piece of persistent data is created several
times, then its causes are merged and the events that read the
persistent data can inherit any of them disjunctively.  Instead, when read arcs
are used, their event structure semantics records the exact causes.

The goal of this paper is to extend Winskel's construction to Petri
nets with persistence, as defined
in~\cite{DBLP:conf/ccs/CrazzolaraW01,CrazzolaraW05}, and to understand
what is the right sub-category of general event structures to
exploit. Surely, prime event structures are not expressive enough,
as they allow for a unique set of minimal causes for each
event. Instead, as discussed above, the presence of persistent data
leads to events with multiple sets of minimal causes.
Consider, e.g., the net with persistence in Figure~\ref{fig:running},
which will serve as a running example. As usual, places and
transitions are represented by circles and boxes,
respectively. Persistent places, like $o$ in this case, are
represented by double circles. Intuitively, transition $c$ is
immediately enabled and can be fired many times (sequentially).  As
soon as $a$ or $b$ fires, place $o$ becomes marked. The number of
tokens in $o$ as well as their causal histories are irrelevant: once
$o$ is marked, it will remain so forever. At this point $d$ and $e$
become (concurrently) enabled.  The firing of $d$ disables $c$ but not
$e$ since the token never leaves the persistent place $o$.

\begin{figure}[t]
$$
 \xymatrix@R=.8pc@C=1.2pc{
 &
 \drawmarkedplace\ar[d]
 \nameplaceright p
 & &
 \drawmarkedplace\ar[d]
 \nameplaceright q
 \\
 &
 \drawtrans a \ar[dr] 
 & &
 \drawtrans b \ar[dl] 
 \\
 \drawtrans c \ar@/^/[r]
 &
 \drawmarkedplace\ar[d] \ar@/^/[l]
 \nameplaceright r
 &
 \drawpersistentplace\ar[dl] \ar[dr]
 \nameplaceright o
 &
 \drawmarkedplace\ar[d] 
 \nameplaceright s
 \\
 &
 \drawtrans d \ar[d]
 & &
 \drawtrans e \ar[d] 
 \\
 &
 \drawplace
  \nameplaceright t
 & &
 \drawplace
  \nameplaceright u
 }
$$
\caption{Running example}\label{fig:running}
\end{figure}
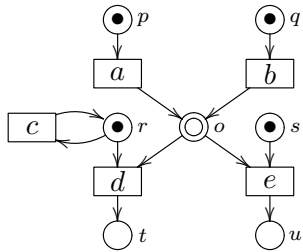

In a recent work~\cite{BCG:DESF} dealing with the semantics of
formalisms with fusion and with the corresponding phenomenon of
disjunctive causality, the domains of configurations
of general event structures are characterised by resorting to a weak
form of prime algebraicity.
\emph{Connected event structures} are identified as canonical representations
of general (possibly non-stable) event structures, in the same way as
prime event structures are representations of stable event structures.
In particular, the equivalence between the category of prime algebraic
domains and the one of prime event structures is generalised to an
equivalence between the category of weak prime domains and that of
connected event structures.
In a connected event structure, an event may have multiple causal
histories, but they are required to be ``connected'', namely they must
not be separable in two classes of pairwise incompatible causal
histories. The idea is that, if in a general event structure an event
has conflicting classes of causal histories, then it should split in
several copies when generating the corresponding connected event
structure.
Indeed, as discussed in~\cite{BCG:DESF}, connected event structures
can be alternatively presented as prime event structures where some
events (intuitively, those having different but not incompatible causal 
histories)
are
deemed equivalent. This establishes a close connection with the work
in~\cite{VW:SPC}, where, in order to model strategies with disjunctive
causes, the authors deal with \emph{prime event structures with
  equivalence}.

In this paper we rely on the aforementioned work. A major role is
played by a weakening of the connectedness property for event
structures, referred to as \emph{local connectedness}. The underlying
intuition is as follows. A causal history can be seen as a conjunction
of its events, thus an event with different causal histories is
enabled by a disjunction of conjunctions of events.  Connectedness
amounts to the fact that the various conjuncts cannot be split in
conflicting subclasses.
Moving to Petri nets with persistence, a persistent place
can be seen as the disjunction of all events that can fill the
place. In turn, an event needs all places in its pre-set to be filled,
hence it is enabled by a conjunction of disjunctions.
The property of local connectedness roughly amounts to the requirements that 
the different ways of enabling a persistent place cannot be separated into conflicting classes.  If this were possible, one should split the place in different copies, one for each class.
The notion of local connectedness lifts to event structures and we
show that Winskel's chain of coreflections can be generalised to link
the category of Petri nets with persistence to that of locally
connected event structures. The latter, in turn, coreflects into the
category of connected event structures.

The result can be read from two perspectives. From one viewpoint,
where the construction defines the event structure associated with a
net, it characterises the ``right'' concurrent semantics for dealing
with formalisms that handle persistent information. The interesting
bit is that the notion of connectedness from~\cite{BCG:DESF} is
relaxed here to local connectedness. From the second
viewpoint, the construction builds a standard net that is the best
representative for the (locally connected) event structure at
hand. The latter is a more interesting viewpoint, because: (i)~it shows
that Petri nets with persistence are expressive enough to account for
disjunctive causes, advancing towards the solution of a long-standing open question
about finding the right computational model for general event
structures; %
and (ii)~it confirms that Petri
nets offer the ideal playground to experiment with concurrency
features.

\subsection{Structure of the paper}
In Section~\ref{sec:es} we review the preliminaries on connected event
structures and we introduce the original class of locally connected
event structures.
In Section~\ref{sec:pnets} we introduce Petri nets with persistence (p-nets) and the corresponding category.
In Section~\ref{sec:unf} we introduce occurrence p-nets and define a coreflection between the category of p-nets and the one of occurrence p-nets. 
The right-adjoint of the coreflection is the unfolding construction
that accounts for the description of concurrent computations of a
p-net.
Technically, it is defined in two steps, going through a category of occurrence p-nets with equivalence.
In Section~\ref{sec:les} we establish a coreflection between the
category of p-nets and the one of locally connected event structures.
Some concluding remarks are in Section~\ref{sec:conc}.

\section{Event structures and (local) connectedness}
\label{sec:es}

In this section we review the basics of event structures~\cite{Win:ES}
and the notion of connected event structure from~\cite{BCG:DESF}.
Then we single out a wider class of event structures, referred to as
\emph{locally connected}, that will play a pivotal role in the paper.

We start by recalling the notion of event structure with binary conflict~\cite{Win:ES}. 
In the following, for $m,n \in \mathbb{N}$, we denote by
$\interval[m]{n}$ the set $\{ m, m+1, \ldots, n\}$.
Also, given a set $E$, we denote by $\pow{E}$ the powerset of $E$ and by
$\powfin{E}$ the set of finite subsets of $E$.

\begin{defi}[event structure]
  \label{de:es}
  An \emph{event structure} ({\evstr} for short) is a tuple
  $\langle E, \vdash, \# \rangle$ such that
  \begin{itemize}
  \item $E$ is a set of events;
  \item ${\vdash} \subseteq \powfin{E} \times E$ is the \emph{enabling}
    relation satisfying $X \vdash e$ and $X \subseteq Y$
    implies $Y \vdash e$;
  \item $\# \subseteq E \times E$ is the \emph{conflict} relation.
  \end{itemize}
  Two events $e, e' \in E$ are \emph{consistent}, written
  $e \conn{} e'$, if $\neg (e \# e')$. A subset $X \subseteq E$ is
  \emph{consistent} if $e \conn{} e'$ for all $e, e' \in X$.
\end{defi}

An {\evstr} $\langle E, \vdash, \# \rangle$ is often
denoted simply by $E$.  Computations are captured in the form of
configurations.

\begin{defi}[configuration, live {\evstr}, concurrent events]
  A \emph{configuration} of an {\evstr} $E$ is a consistent
  $C \subseteq E$ which is \emph{secured}, i.e., for all $e \in C$
  there are $e_1, \ldots, e_n \in C$ with $e_n =e$ such that
  $\{ e_1, \ldots, e_{k-1} \} \vdash e_k$ for all $k \in \interval{n}$
  (in particular, $\emptyset \vdash e_1$). The set of configurations
  of an {\evstr} $E$ is denoted by $\conf{E}$.
  An {\evstr} is \emph{live} if  it has no self-conflicts, i.e., for all $e \in E$ we
  have $\neg (e \# e)$, and conflict is saturated, i.e., for all
  $e, e' \in E$, if $\neg (e \# e')$ then there is $C \in \conf{E}$
  such that $\{e,e'\} \subseteq C$.
  Two events $e, e' \in E$ are \emph{concurrent} if they are
  consistent ($e \conn{} e'$) and there is
  $C \in \conf{E}$ such that $C \vdash e$ and $C \vdash e'$.
\end{defi}

Thus in a live {\evstr} conflict is saturated, a property that
corresponds to inheritance of conflict in prime {\evstr}s, and
each event is executable.

\begin{rem}
  In the paper we restrict to live {\evstr}. Hence the
  qualification \emph{live} is omitted. 
\end{rem}

The class of {\evstr} can be turned into a category.

\begin{defi}[category of {\evstr}]
  \label{de:es-morphism}
  A morphism of {\evstr} $f : \langle E_1, \vdash_1, \#_1 \rangle
  \to \langle E_2, \vdash_2, \#_2 \rangle$
  is a partial function $f : E_1 \rightharpoonup E_2$ such that for all
  $C_1 \in \conf{E_1}$ and $e_1, e_1' \in E_1$ with $f(e_1)$, $f(e_1')$ defined
  \begin{itemize}
  \item if $f(e_1)\, \#_2\, f(e_1')$ then $e_1\,\#_1\, e_1'$
  \item if $f(e_1) = f(e_1')$ and  $e_1 \neq e_1'$ then $e_1\, \#_1\, e_1'$;
  \item if $C_1 \vdash_1 e_1$ then $f(C_1) \vdash_2 f(e_1)$.
  \end{itemize}
  We denote by $\es$ the category of {\evstr} and their morphisms. 
\end{defi}

Since the enabling predicate is over finite sets of events, we can
consider minimal sets of events enabling a given one.

\begin{defi}[minimal enabling, causality]
  Let $\langle E, \vdash, \# \rangle$ be an {\evstr}. Given $e \in E$
  and $C \in \conf{E}$ such that $C \vdash e$ we say that $C$ is a
  \emph{minimal enabling} of $e$, and write $C \vdash_0 e$, when for
  any configuration $C' \subseteq C$, if $C' \vdash e$ then
  $C' = C$. We denote by $\hist{e} = \{ C \mid C \vdash_0 e \}$  the
  set of minimal enablings of event $e$. We write $e < e'$ if
  $e \in C$ for all $C\in \hist{e'}$.
\end{defi}

The configurations of an {\evstr}, ordered by subset inclusion,
form a partial order that is characterised in~\cite{BCG:DESF} as a weak prime
algebraic domain, i.e., a coherent finitary partial order where each
element is the join of elements satisfying a weak notion of primality.
The relation is formalised as a coreflection between the category $\es$ and a category $\WDom$ of weak prime domains.
A subclass of {\evstr}s can be identified, called \emph{connected
{\evstr}}, that represents the exact counterpart of weak prime domains,
in the same way as prime {\evstr}s are the counterpart of prime
algebraic domains~\cite{Win:ES}.

\begin{defi}[connected {\evstr}]
\label{def:connectedEs}
  Let $C, C' \in \hist{e}$.
We write $C \conn{e} C'$ if $C \cup C' \cup \{ e \}$
  is consistent and we denote by $\conn{e}^*$
  the transitive closure of the relation $\conn{e}$.
  An {\evstr} is \emph{connected} if whenever $C, C' \in \hist{e}$
  then $C \conn{e}^* C'$.
  The full subcategory of $\es$ having connected {\evstr} as objects
  is denoted by $\ces$.
\end{defi}

The category $\ces$ is equivalent to the category $\WDom$ of weak prime domains 
and thus it coreflects in $\es$.

\begin{prop}[coreflection between $\es$ and $\ces$~\cite{BCG:DESF}]
  \label{pr:coreflection-es-ces}
  The inclusion functor $\mathcal{I} : \ces \to \es$ admits a right
  adjoint $\mathcal{C} : \es \to \ces$ establishing a coreflection.
\end{prop}

As mentioned in  the introduction, the concurrent semantics of Petri nets with persistence will be given in terms of connected {\evstr} through a chain of transformations that first unfolds the net into an acyclic net and then abstracts it to an {\evstr}.  

The connectedness condition has a natural logical interpretation. Given an event $e$, we can capture its dependencies by stating that $e$ is caused by the disjunction of its minimal enablings, where each minimal enabling can be seen in turn as a 
conjunction of events, namely by  
$\bigvee_{C \in \hist{e}} \bigwedge C$.
A conflict $e' \# e''$ can be encoded as $\neg e' \lor \neg e''$. In this view, connectedness amounts to the impossibility of partitioning $\hist{e}$ in two subsets inducing mutually exclusive minimal causal histories, i.e., we cannot decompose $\hist{e} = H_1 \uplus H_2$ in a way that 
$\neg (\bigvee_{C \in H_1} \bigwedge (C \cup \{e\})) \ \lor\ \neg (\bigvee_{C \in H_2} \bigwedge (C \cup\{e\}))$. 
If this happened, to recover connectedness we should split event $e$ in two events $e_1$ and $e_2$, with $\hist{e_i} = H_i$ for $i \in \{ 1,2\}$.

At the level of nets, events correspond to transitions.  
Each transition requires that all the
(possibly persistent) places in its pre-set are filled in order to be enabled, hence it is
enabled by a \emph{conjunction} of places. In turn, each
persistent place can be seen as the \emph{disjunction} of the
transitions in its pre-set.  
Summing up, at the level of
nets we can represent conjunctions of disjunctions of events, exploiting persistence. The
natural choice, when working with nets, will be to impose the
connectedness condition locally to each disjunct. This results in a
property weaker than the ``global'' connectedness from
Definition~\ref{def:connectedEs}. For this reason, the extraction of a
connected {\evstr} from a net will pass through an intermediate class
of {\evstr} that we call \emph{locally connected}. We next formalise this
idea.

\begin{defi}[{\evstr} in disjunctive form]
  Let $E$ be an {\evstr}. Given $e \in E$, a
  \emph{disjunct} of $e$ is a minimal set $X \subseteq E$ such that
  $X \cap C \neq \emptyset$ for all
  $C \in \hist{e}$.
  It is \emph{connected} if for all $e, e' \in X$ there exists
  $n\geq 1$ and $e_1,e_2,...,e_n\in X$ such that
  $e=e_1\conn{} e_2 \conn{} \cdots \conn{} e_n = e'$.

  A \emph{covering} of $e$ is a set of disjuncts
  $D \subseteq \pow{E}$ such that for any
  $C \in \conf{E}$, if $C \cap X \neq \emptyset$ for all $X \in D$
  then $C \vdash e$.
\end{defi}

Intuitively, $D$ is a covering of an event $e$ whenever condition
$\bigwedge_{X \in D} \bigvee X$ is necessary and sufficient 
for enabling $e$, i.e., it is logically equivalent to the disjunction
of the minimal enablings $\bigvee_{C \in \hist{e}} \bigwedge C$. A
disjunct $X$ of $e$ is connected whenever it cannot be partitioned as
$X = X_1 \uplus X_2$ with 
$\neg (\bigvee X_1) \lor \neg (\bigvee X_2)$.  
Intuitively, we are moving from a $\vee$-$\wedge$ form of the
dependencies to a $\wedge$-$\vee$ form, and transferring the
connectedness condition from the outer to the inner disjunctions.
Expressing dependencies
as a conjunction of disjunctive causes
makes it natural to associate a net with persistent places with the {\evstr}: each event $e$ becomes a transition and each disjunct $X$ of $e$ corresponds to a persistent place $s$ in the pre-set of $e$, filled by the events in $X$. 
The guarantee that the
disjunctive causes cannot be split into inconsistent subsets will provide a form of canonicity to the construction.

Note that an event enabled by the empty set has no disjuncts (the empty set, which as a disjunct would correspond to ``true'', is not admitted).

\begin{defi}[locally connected {\evstr}]
  \label{de:loc-conn-es}
  An {\evstr} $E$ is \emph{locally connected} if for all
  $e \in E$ there exists a covering $D$ such that any $X \in D$ is
  connected.
  We denote by $\les$
  the full subcategory of $\es$ having locally connected event
  structures as objects.
\end{defi}

It can be easily shown that connectedness implies local connectedness. We first observe that in an {\evstr} every event admits a covering, which is the set of all its disjuncts.

\begin{lem}[coverings always exist]
  \label{le:cover}
  Let $E$ be an {\evstr} and $e \in E$. Then
  $D_e = \{ X \mid X \subseteq E \land\ \mbox{$X$ is a disjunct of $e$}\}$
  is a covering of $e$.
\end{lem}

\begin{proof}
  Let $C \in \conf{E}$ be a configuration such that
  $C \cap X \neq \emptyset$ for all $X \in D_e$, and suppose by absurd
  that $C \not \vdash e$.  This means that for all $C' \in \hist{e}$
  there is an event $e_{C'} \in C' \setminus C$. Let $X'$ be a minimal
  subset of $\{ e_{C'} \mid C' \in \hist{e}\}$ such that
  $X' \cap C' \not = \emptyset$ for all $C' \in \hist{e}$: then
  clearly $X'$ is a disjunct in $D_e$, but $X' \cap C = \emptyset$,
  yielding a contradiction.
\end{proof}

\begin{lem}[connectedness vs local connectedness]
  \label{le:ces-vs-les}
  Let $E$ be an {\evstr}. If $E$ is connected then it is locally connected.
\end{lem}

\begin{proof}
  Let $E$ be a connected {\evstr} and let $e \in E$ be an event. 
  Let
  $X$ be any disjunct of $e$ and consider $e_1,e_2 \in X$.
  By minimality of a disjunct, we deduce that there must be
  $C_1, C_2 \in \hist{e}$ such that $e_i \in C_i$, for
  $i \in \{1,2\}$. Since $E$ is connected, we know that
  $C_1 \conn{e}^* C_2$. Then we can prove that $e_1 \conn{}^* e_2$ by
  induction on the length of the chain of consistency
  $C_1 \conn{e}^* C_2$. The base case is trivial. For the inductive
  step, assume that $C_1 \conn{e}^* C_1' \conn{e} C_2$. Take any
  $e_1' \in C_1'$. By inductive hypothesis we know that
  $e_1 \conn{}^* e_1'$. Moreover, since $C_1' \conn{e} C_2$ we deduce
  $e_1' \conn{} e_2$. Thus we conclude $e_1 \conn{}^* e_2$ as desired.
\end{proof}

The above result shows that $\ces$ is a full subcategory of $\les$. Hence the coreflection between $\es$ and $\ces$ restricts to a coreflection between $\les$ and $\ces$.

\begin{prop}[coreflection between $\les$ and $\ces$]
  \label{pr:coreflection-les-ces}
  The inclusion functor $\mathcal{I} : \ces \to \les$ is left adjoint
  of the restriction $\mathcal{C}_{|\les} : \les \to \ces$,
  establishing a coreflection.
\end{prop}

\begin{proof}
  Immediate consequence of Proposition~\ref{pr:coreflection-es-ces}
  and Lemma~\ref{le:ces-vs-les}.
\end{proof}

Local connectedness is strictly weaker than connectedness. Let 
$E = \{ a, b, c, d, e \}$ be the {\evstr} with
$\emptyset \vdash_0 a$, $\emptyset \vdash_0 b$,
$\emptyset \vdash_0 c$, $\emptyset \vdash_0 d$, $a \# c$, $b \# d$,
$\{ a,b \} \vdash_0 e$ and $\{ c,d \} \vdash_0 e$. This is not
so, since $\hist{e} = \{ \{a,b\}, \{c,d\}\}$ and $\{ a, b \} \conn{e} \{ c,d\}$
does not hold. It is locally
connected, since $\{ \{a,d\}, \{ b, c\}\}$ is a covering of $e$ and
the disjuncts $\{a,d\}$ and $\{b,c\}$ are connected since
$a \conn{} d$ and $b \conn{} c$.
Logically, the cause of $e$ in $\vee$-$\wedge$ form is $(a \wedge b) \vee (c \wedge d)$. It is not connected since $\neg (a \wedge b) \vee \neg (c \wedge d)$ (and thus $\neg (a \wedge b \wedge e) \vee \neg (c \wedge d \wedge e)$).
If we put the causes in $\wedge$-$\vee$ form we get $(a \vee d) \wedge (b \vee c)$, where neither $\neg a \vee \neg d$ nor $\neg b \vee \neg c$, whence local connectedness.

Dealing with (locally) connected {\evstr} will play an essential role
for establishing a coreflection between occurrence nets with
persistence and {\evstr}s (see Theorem~\ref{th:OccEs}). At an
intuitive level, it ensures that or-causality is preserved along
morphisms and cannot be transformed in ordinary causality. Instead, consider
for instance the {\evstr} $E_1$ and $E_2$ defined as follows
\begin{itemize}
\item $E_1 = \{ a_1, a_2, b \}$ with $\emptyset \vdash_0 a_i$,
  $\{a_i\} \vdash_0 b$ for $i \in \{ 1,2\}$ and $a_1 \# a_2$;

\item $E_2 = \{ a, b\}$ with $\emptyset \vdash_0 a$ and $a \vdash_0 b$

\end{itemize}
The {\evstr} $E_1$ is not locally-connected, since $\{a_1,a_2\}$ is the only disjunct for $b$ and it is not connected. It is easy to realise that the mapping $f : E_1 \to E_2$ defined by
$f(a_1) = f(a_2) = a$ and $f(b) =b$ is an {\evstr} morphism: merging $a_1$ and $a_2$ the or-causality of $b$ is transformed into a proper causality $a < b$.

\section{Nets with Persistence}
\label{sec:pnets}

In this section we introduce Petri nets with persistence. Since the state of a
Petri net will be seen as a multiset, i.e., an element of a suitably defined monoid, we start with some notation on sets and monoids. The irrelevance of the number of tokens in persistent places is modelled by some form of idempotency in the monoid.

Recall that, given a set $X$, we denote by $\pow{X}$ the powerset of $X$ and by
$\powfin{X}$ the set of finite subsets of $X$.
We denote by $\mon{X} = \{ u \mid u : X \to \nat \}$ the commutative
monoid of multisets on $X$ with the operation $\oplus$ and identity
$\emptyset$ defined in an obvious way. Elements of $\mon{X}$ are often
represented by formal sums $u = \bigoplus_{x \in X} u(x) \cdot x$.
Given $u \in \mon{X}$ we denote by $\set{u}$ the underlying set $\{ x
\mid x \in X\ \land\ u(x) > 0 \}$. We write $x \in u$ for $x \in
\set{u}$ and we say that $u$ is \emph{finite} if $\set{u}$ is
finite. We identify a set $u \in \pow{X}$ and the ``corresponding''
multiset $\bigoplus_{x \in u} 1 \cdot x$.
Given $u,u'\in \mon{X}$ we say that $u$ \emph{covers} $u'$, written
$u'\subseteq u$, if there is $u''\in \mon{X}$ such that $u = u'\oplus
u''$.
For $u, u' \in \mon{X}$ we write $u \cap u'$ for the largest $u''$ such that $u'' \subseteq u$ and $u'' \subseteq u'$.

When $X$ has a chosen subset $\pe{X} \subseteq X$, the
commutative monoid with idempotency is $\mon{(X,\pe{X})} = \{ u \mid u
: X \to \nat\ \land\ \forall x \in \pe{X}.\, u(x) \leq 1 \}$. Elements
are still seen as formal sums $u = \bigoplus_{x \in X} u(x) \cdot
x$, with an idempotency axiom $x \oplus x = x$ for any $x \in \pe{X}$.
As before, given $u,u'\in \mon{(X,\pe{X})}$ we say that $u$ \emph{covers} $u'$, written
$u'\subseteq u$, if there is $u''\in \mon{X}$ such that $u = u'\oplus u''$.
Note however that due to idempotency there can be several $u''$ such that $u = u'\oplus u''$.
More precisely, the set $\{ u'' \mid u = u'\oplus u''\}$ forms a lattice with respect to $\subseteq$ and we write $u\ominus u'$ to denote the top element of the lattice.
For example, when $a\in X\setminus \pe{X}$ and $b,c\in \pe{X}$, although $(a\oplus b)\oplus(a\oplus c) = 2a\oplus b\oplus c$ we have $(2a\oplus b\oplus c) \ominus (a\oplus b)= a\oplus b\oplus c$.

Given sets $X$, $Y$, a monoid homomorphism $f : \mon{X} \to \mon{Y}$
is called \emph{finitary} if for all $x \in X$, $f(x)$ is finite.  A
function $f : X \to \mon{Y}$ that is finitary (i.e., $f(x)$ finite
for all $x \in X$) can be extended to a finitary monoid homomorphism
denoted $\mon{f} : \mon{X} \to \mon{Y}$, defined by
$\mon{f}(\bigoplus_{x \in X} n_x \cdot x) = \bigoplus_{x \in X} n_x
\cdot f(x)$.
Analogous notions can be defined for partial functions and when the
target or both the source and target are monoids with idempotency.

\begin{defi}[net with persistence]
  An (unweighted marked P/T Petri) \emph{net with persistence}
  (\emph{p-net}, for short) is a tuple
  $N = (S, \pe{S},T, \delta_0, \delta_1, u_0)$, where $S$ is a set of
  \emph{places}, $\pe{S} \subseteq S$ is the set of \emph{persistent} places,
  $T$ is a set of \emph{transitions},
  $\delta_0, \delta_1 : T \to \pow{S}$ are functions assigning sets
  called pre-set and post-set, respectively, to each transition and
  $u_0 \subseteq \npe{S}$ is the \emph{initial marking}, where $\npe{S}$
  is the set of non-persistent places
 $\npe{S} = S \setminus \pe{S}$.
\end{defi}

Given a finite multiset of transitions $v \in \mon{T}$ we write
$\pre{v}$ and $\post{v}$ for $\mon{\delta_0}(v)$ and
$\mon{\delta_1}(v)$.
 Given a place $s \in S$ we also write $\pre{s}$ for
$\{ t \mid s \in \post{t} \}$ and $\post{s}$ for  $\{ t \mid s \in \pre{t}
\}$. 

Hereafter, for any p-net $N$ we assume $N = (S, \pe{S},T, \delta_0,
\delta_1, u_0)$, with subscripts and superscripts carrying over the
names of the components.
Note that we work with nets that are not weighted (pre- and post-sets
of transitions are sets, rather than multisets). The results could be
trivially extended, as in the ordinary case, to semi-weighted nets
(where the pre-set can be a proper multiset). The restriction is
adopted to ease the presentation.

The state of a p-net $N$ is represented by some $u \in
\mon{(S,\pe{S})}$, called a \emph{marking} of $N$. A transition $t$ is
\emph{enabled} by a marking $u$ if its pre-set is covered by $u$,
i.e. if $\pre{t} \subseteq u$.
In this case, $t$ can be fired. The firing of $t$ consumes the tokens in the
non-persistent places in the pre-set, leaves untouched the tokens in the
persistent places and produces the tokens in the post-set. More generally,
this applies to finite multisets of transitions. Formally, given a
finite multiset $v \in \mon{T}$ and a marking $u \in
\mon{(S,\pe{S})}$, if 
$\pre{v} \subseteq u$, the firing rule is 
$u\fire{v}(u\ominus \pre{v})\oplus\post{v}$.

A marking $u \in \mon{S}$ is \emph{reachable} if there exists a firing
sequence $u_0 \fire{t_1} u_1 \fire{t_2} \ldots \fire{t_n} u$ from the
initial marking to $u$.
The p-net is \emph{safe} if every reachable marking is a set.

Observe now that in the initial marking only non-persistent places can be
marked.  %
Indeed, if a persistent place $p$
were marked in the initial marking then its presence would be
essentially useless, since any reachable marking would contain one token in
$p$. Therefore removing $p$ and all the incoming and outgoing arcs
would lead to an equivalent net that can perform the same firing
sequences and such that a marking $u$ is reachable if and only if 
the marking $u\oplus p$ was
reachable in the original net.

\begin{exa}
Let us consider the p-net in Figure~\ref{fig:running} of our running
example, whose initial marking is $u_0 = p \oplus q \oplus r \oplus
s$. A sample firing sequence is
$$
u_0 
\ \fire{a}\ q \oplus r \oplus s \oplus o
\ \fire{b}\ r \oplus s \oplus o
\ \fire{c}\ r \oplus s \oplus o
\ \fire{d}\ t \oplus s \oplus o
\ \fire{e}\ t \oplus u \oplus o
$$
Note that once the token in the persistent place $o$ is produced by the
firing of $a$, it is never possible to remove it, not even firing $d$ or
$e$. Also note that after $b$ fires, the multiplicity of $o$ remains
$1$ due to the idempotency axiom on persistent places. It is immediate
to check that the p-net is safe, as all reachable markings are 
sets.
\end{exa}

\begin{rem}
  Our presentation of p-nets slightly differs from the original one
  in~\cite{CrazzolaraW05}.  There: (i)~arcs carry weight $1$ or
  $\infty$; (ii)~an arc has weight $\infty$ if and only if it goes
  from a transition to a persistent place; (iii)~markings allows
  having either $0$ or infinitely many tokens in each persistent
  place; (iv)~even if a firing removes finitely many tokens from
  persistent place, there remain infinitely many tokens available.  As
  each marking associates only one bit of information with each
  persistent place, here we find technically more convenient to
  represent marked persistent places by assigning them multiplicity
  $1$ and by exploiting idempotency to capture the fact that when infinitely
  many tokens are added to a marked persistent place there still are
  infinitely many tokens (indistinguishable one from the other).
\end{rem}

\begin{rem}
We use the term ``net with persistence'' to avoid confusion with the
notion of ``persistent net'' in the literature 
(see e.g.~\cite{DBLP:journals/jacm/LandweberR78}),
a behavioural property defined as follows: A net is
persistent if whenever $u \fire{t_1}$ and $u \fire{t_2}$ for a
reachable marking $u$, then $u \fire{t_1\oplus t_2}$.
\end{rem}

When dealing with unfolding and, more generally, with the causal semantics
of ordinary Petri nets, it is a standard constraint to assume that transitions have a
non-empty pre-set. This
avoids the unnatural
possibility of firing infinitely many copies of the same transition in
parallel. For p-nets, this generalises to the requirement
that each transition consumes tokens from at least one non-persistent
place.
Additionally, since a persistent place will
never be emptied once it is  filled with a token, 
whenever a persistent place $s$ is in the post-set
of a transition $t$ it is quite natural to forbid the presence of an
additional path from $t$ to $s$. This property is formalised by using the flow
relation $\struct_N$ for a net $N$, defined, for all
$x, y \in S \cup T$, by $x \struct_N y$ if $x \in \pre{y}$.

\begin{defi}[well-formed net]
  \label{de:well-formed}
  A p-net $N$ is \emph{well-formed} if for all $t \in T$,
  $\delta_0(t) \not \in \mon{\pe{S}}$ (\emph{t-restrictedness}) and
  for all $t \in T$, $s \in \pe{S}$, if  $t \struct_N s$ then $t \not \struct_N^n s$ for $n \geq 2$  (\emph{irredundancy}).
\end{defi}

Observe that, in particular, whenever a p-net is irredundant, it does not include cycles over persistent places, i.e., for any $s \in \pe{S}$, it is not the case that $s \struct_N^+ s$. Hereafter all p-nets will be tacitly assumed to be well-formed.

The notion of p-net morphism naturally arises from an algebraic view,
where places and transitions play the role of sorts and operators.

\begin{defi}[p-net morphism]
  \label{de:pnet-morphism}
  A \emph{p-net morphism}
  $f = \langle \pl{f},\tr{f}\rangle : N \to N'$ is a
  pair where 
  \begin{enumerate}
  \item
    \label{de:pnet-morphism:1}
    $\pl{f} : \mon{(S,\pe{S})} \to \mon{(S',\pe{S'})}$ is a
    finitary monoid homomorphism such that for $s \in \npe{S}$,
    $\pl{f}(s) \in \mon{\npe{S'}}$ and the initial marking is preserved, 
    i.e.~ $\pl{f}(u_0) = u'_0$;
    
  \item
    \label{de:pnet-morphism:2}
    $\tr{f} : T \rightharpoonup T'$ is a partial function such
    that for all finite $v \in \mon{T}$,
    $\pre{(\mon{\tr{f}}(v))} = \pl{f}(\pre{v})$ and
    $\post{(\mon{\tr{f}}(v))} = \pl{f}(\post{v})$;

  \item
    \label{de:pnet-morphism:3}
    for all $t \in T$ and $s_1, s_2 \in \pre{t}$ or $s_1, s_2 \in \post{t}$, if
    $\pl{f}(s_1) \cap \pl{f}(s_2) \neq \emptyset$ then $s_1 = s_2$.  
  \end{enumerate}
  The category of p-nets (as objects) and their morphisms (as arrows)
  is denoted by $\pn$.
\end{defi}

Observe that by the fact that
$f_s : \mon{(S,\pe{S})} \to \mon{(S',\pe{S'})}$ is a monoid homomorphism and condition~(\ref{de:pnet-morphism:1}) we
automatically get that $s \in \pe{S}$ implies $\pl{f}(s) \in \mon{\pe{S'}}$.
Moreover, for all
$u_1, u_2 \in \mon{(S,\pe{S})}$, with $u_2 \subseteq u_1$, it holds
that $\pl{f}(u_1 \ominus u_2) = \pl{f}(u_1) \ominus \pl{f}(u_2)$.
Condition~(\ref{de:pnet-morphism:2}) amounts to require that for any
$t \in T$, if $\tr{f}(t)$ is defined then
$\pre{{\tr{f}}(t)} = \pl{f}(\pre{t})$ and
$\post{{\tr{f}}(t)} = \pl{f}(\post{t})$, and
$\pl{f}(\pre{t}) = \pl{f}(\post{t}) = \emptyset$ otherwise.
Finally, condition~(\ref{de:pnet-morphism:3}) imposes injectivity of the morphism on the pre-set
and post-set of each transition, i.e., the morphism should not mix
places in the pre-set (and in the post-set) of the same transition.
This is automatically satisfied for non-persistent places, 
but it
could be violated by persistent places (due to idempotency), hence we
require it explicitly.

In the sequel, when the meaning is clear from the context, we often omit the
subscripts from the morphism components, thus writing $f$ instead of
$\pl{f}$ and $\tr{f}$.

\begin{lem}[p-net morphisms are simulations]
  \label{lemma:simulation}
  Let $f : N \to N'$ be a p-net morphism.  If $u \fire{v} u'$ is a
  firing in $N$ then $f(u) \fire{\mon{f}(v)} f(u')$ is a firing in
  $N'$.
\end{lem}

\begin{proof}
  Since $f$ is a p-net morphism we have
  $\pre{(\mon{f}(v))} = f(\pre{v})$ and
  $\post{(\mon{f}(v))} = f(\post{v})$.  Since $u \fire{v} u'$, it must
  be the case that $u = u'' \oplus \pre{v}$ for some $u''$, from which
  $f(u) = f(u'') \oplus f(\pre{v})$, i.e., $f(\pre{v}) \subseteq f(u)$
  and thus $\mon{f}(v)$ is enabled at $f(u)$.
  Therefore there is a firing
  \begin{center}
    $f(u) \fire{\mon{f}(v)} (f(u) \ominus \pre{(\mon{f}(v))}) \oplus
    \post{(\mon{f}(v))}$
  \end{center}
  and the target
  $(f(u) \ominus \pre{(\mon{f}(v))}) \oplus \post{(\mon{f}(v))}= (f(u)
  \ominus f(\pre{v})) \oplus f(\post{v}) = f((u \ominus \pre{v})
  \oplus \post{v}) = f(u')$, as it was required.
\end{proof}

\section{Unfolding Nets with Persistence}
\label{sec:unf}

In this section we show how a p-net can be unfolded to a suitably
defined occurrence p-net that represents all possible occurrences of
firing of transitions and their dependencies. We first introduce the
class of occurrence p-nets, where each transition can occur at most
once in a computation, but possibly with different disjunctive
causes. We next observe that occurrence p-nets can be equivalently
presented by forcing each transition 
to have a uniquely determined set of
causes and using an equivalence between transitions to account for
disjunctive causes. Finally, we present the unfolding construction for
a p-net, that works in two steps. First a p-net is unfolded into an
occurrence p-net without backward conflicts, where each item has a
uniquely determined causal history. The possibility of generating a
token in the same persistent place with different histories is
captured by means of an equivalence relation on places and
transitions. Then the actual unfolding is obtained as the quotient of the
pre-unfolding with respect to the equivalence relation.

\subsection{Occurrence p-nets}

In order to single out the class of occurrence p-nets, we start by defining the possible dependencies in a p-net.

\begin{defi}[enabling, conflict, causality, dependence, concurrency]\label{def:mainrels}
\label{def:dependence}
  Let $N$ be a p-net. 
\begin{itemize}
\item  Enabling, written $\vdash_N$, is defined by letting, for $X \in \powfin{T}$
  and
  $t \in T$, $X \vdash_N t$ if for all $s \in \pre{t}$ either
  $s \in u_0$
  or there exists $t' \in X$ such that
  $s \in \post{t'}$.
\item Conflict $\#_N \subseteq \pow{S \cup T}$ is the least
  set-relation, closed by superset, 
  defined by

    \begin{enumerate}
    
  \item[(a)] if $t \neq t'$ and $\pre{t} \cap \pre{t'} \not\subseteq \pe{S}$,
    then $\#_N \{ t,  t'\}$;
  
  \item[(b)] if $\pre{s} \neq \emptyset$ and $\#_N (X \cup \{t\})$ for all
   $t \in \pre{s}$, then
   $\#_N (X \cup \{s \})$;

 \item[(c)] if $\#_N (X \cup \{ s \})$ and $s \in \pre{t}$, then
    $\#_N (X \cup \{ t\})$.

  \end{enumerate}
  We will often write $x \#_N x'$ instead of $\#_N \{ x,x'\}$. 
  We say that $X$ is
 \emph{consistent} if it is not in conflict. In particular, binary consistency is denoted $\conn{}_N$, i.e., $x \conn{}_N x'$ when $\neg(x \#_N x')$. 
 We say that $X$ is \emph{connected} by $\conn{}_N$ if for all $x,x'\in X$ there exists $n\geq 1$ and $x_1,x_2,...,x_n\in X$ such that $x=x_1 \conn{}_N x_2 \conn{}_N \cdots \conn{}_N x_n = x'$.
\item
  Causality $\leq_N$ is the least transitive and reflexive
  relation $\leq_N$ on $S \cup T$ such that if $s \in \pre{t}$ then
  $s \leq_N t$ and if $\pre{s} = \{t\}$
  then $t \leq_N s$.
  We write $x <_N x'$ when $x \leq_N x'$ and $x\neq x'$.
\end{itemize}
  When the context is clear, we will omit the subscript $N$,  writing 
  $\vdash$, $\#$, $\leq$ and $\conn{}$
  instead of  $\vdash_N$, $\#_N$, $\leq_N$  and $\conn{}_N$.
\end{defi}

Differently from what happens for ordinary nets, in occurrence
p-nets we will allow different ways of enabling the same occurrence of
a transition.
This is because tokens, once generated in persistent places, cannot be
consumed and all tokens in the same place are ``merged'' into one (by
idempotency) in a way that the resulting single token joins all the
different possible causal histories.
For this reason it is convenient to resort to a general notion of
enabling where a transition naturally has several sets of transitions that
allow for its execution.
Observe that if $X \vdash e$ and $X \subseteq X'$ then $X' \vdash
e$. As it happens for {\evstr}s, when
$X \vdash e$, the set $X$ can be inconsistent. We will
later impose a condition forcing each transition to be enabled by at
least one minimal consistent set.

Note that direct conflict is only
binary, since it is determined by the competition on non-persistent
places. However, in order to define properly inheritance of conflict along the enabling relation we need to work with a conflict relation on generic sets.
Consider for instance, the p-net in Figure~\hyperref[fig:combined]{\ref*{fig:combined}a}.
and the set of places $\{ p, q, r\}$. The rules for conflict allow us to deduce that $\# \{ p, q, r \}$, while $\neg x \# y$ for all $x, y \in \{p, q, r\}$. 
Intuitively this happens  because at most two places in the set $\{ p, q, r \}$ can be filled.
The formal derivation can be found in Figure~\hyperref[fig:combined]{\ref*{fig:combined}b}.
The relation between coverability and absence of conflicts will be characterised later via the notion of concurrency, in Lemma~\ref{le:cover-conc-general}.

We also introduce causality: when an item $x$ is a cause of an item
$y$ the intuition is that the presence of $y$ in a
computation implies the presence of $x$, i.e., $x$ is needed to
``enable'' $y$. This will play a role later. However, note that
causality alone would not completely characterise the dependencies in
the p-net as it does not account for disjunctive causes.

We next define the notion of securing sequence
in the context of p-nets.

\begin{defi}[securing sequence]
  \label{de:net-securing-disjunct} 
  Let $N$ be a p-net. Given a transition $t \in T$, a \emph{securing
    sequence} for $t$ in $N$ is a sequence of distinct transitions
  $t_1, t_2, \ldots, t_n=t$ such that for all $i, j \in \interval{n}$,
  $\neg (t_i \# t_j)$ and for all $i \in \interval{n}$,
  $\{t_1 \ldots t_{i-1}\} \vdash t_{i}$.
\end{defi}

A securing sequence for $t$ is intended to represent a firing sequence
that leads to the execution of $t$. This fact will be later
formalised, for occurrence p-nets, in Lemma~\ref{le:conf-exec}.

\begin{figure}[t]
  \centering
  \begin{minipage}{\textwidth}
    \centering
    $$
    \xymatrix@R=1pc@C=.2pc{
      & & & \drawmarkedplace\ar[dlll]\ar[dr]\ar[drrrrr]
      & & & & \drawmarkedplace\ar[drrr]\ar[dl]\ar[dlllll]
      & & & \\
      \drawtrans {a_1} \ar[dr] &
      & \drawtrans {a_2} \ar[dl] &
      & \drawtrans {b_1} \ar[dr] &
      & \drawtrans {b_2} \ar[dl] &
      & \drawtrans {c_1} \ar[dr] &
      & \drawtrans {c_2} \ar[dl] \\
      & \drawpersistentplace \nameplaceright p & & &
      & \drawpersistentplace \nameplaceright q & & &
      & \drawpersistentplace \nameplaceright r &
    }
    $$
    (a)
    % \label{fig:conc-cover}
  \end{minipage}
  \\[5mm]
  \begin{minipage}{\textwidth}
    \centering
  { \small
  $
  \infer
  {\#\{p,q,r\}}
  {
    \infer
    {\# \{p,q,c_1\}}
    {
      \infer
      {\# \{p,b_1,c_1\}}
      {\infer
        {
          \# \{b_1,c_1\}
        }
        {
          \pre{b_1} \cap \pre{c_1} \not\subseteq \pe{S} 
        }
      }
      &
      \infer
      {\# \{p,b_2,c_1\}}
      {
        \infer
        {\# \{a_1,b_2,c_1\}}
        {
          \infer
          {\# \{a_1,c_1\}}
          {
            \pre{a_1} \cap \pre{c_1} \not\subseteq \pe{S} 
          }
        }
        &
        \infer
        {\# \{a_2,b_2,c_1\}}                {
          \infer
          {\# \{a_2,b_2\}}
          {
            \pre{a_2} \cap \pre{b_2} \not\subseteq \pe{S} 
          }
        }        
      }
    }
    &
    \infer
    {\# \{p,q,c_2\}}
    {
      \infer
      {\# \{p,b_1,c_2\}}
      {
        \infer
        {\# \{a_1,b_1,c_2\}}
        {
          \infer
          {\# \{a_1,b_1\}}
          {
            \pre{a_1} \cap \pre{b_1} \not\subseteq \pe{S} 
          }
        }
        &
        \infer
        {\# \{a_2,b_1,c_2\}}
        {
          \infer
          {\# \{a_2,c_2\}}
          {
            \pre{a_2} \cap \pre{c_2} \not\subseteq \pe{S} 
          }
        }      
      }
      &
      \infer
      {\# \{ p, b_2, c_2\}}
      {
        \infer
        {
          \# \{b_2,c_2\}
        }
        {
          \pre{b_2} \cap \pre{c_2} \not\subseteq \pe{S} 
        }
      }
    }
  }
  $}\\[2mm]
  (b)
  % \label{fig:conc-cover-der}
\end{minipage}
\caption{Non-binary conflict occurrence p-nets}
\label{fig:combined}
\end{figure}

Occurrence p-nets can now be defined as a subclass of p-nets. 
For historical reasons, places and transitions of  occurrence p-nets are called \emph{conditions} and \emph{events}, respectively. We will adopt this convention in the rest of the paper, denoting occurrence p-nets by $O$, with components $(B, \pe{B}, E, \gamma_0, \gamma_1, v_0)$, where $B$ is the set of conditions and $E$ the set of events.

\begin{defi}[occurrence p-net]\label{def:occpnet}
\label{def:occ-p-net}
  A structure  $O = (B, \pe{B}, E, \gamma_0, \gamma_1, v_0)$ is an \emph{occurrence p-net} if it is 
  a p-net such that 
  \begin{enumerate}
  \item
the initial marking   $v_0$ satisfies $v_0 = \{ b\in B \mid \pre{b} =\emptyset \}$;        

  \item each event $e \in E$ admits  a securing sequence in $O$;
  \item for all $e, e' \in e$, if $e \neq e'$ then
    $\post{e} \cap \post{e'} \subseteq \pe{B}$;

  \item for all $b \in B$, $\pre{b}$ is connected by $\conn{}_O$.

  \end{enumerate}
An occurrence p-net is \emph{without backward conflicts} when for all
  $e, e' \in E$, if $e \neq e'$ then
  $\post{e} \cap \post{e'} = \emptyset$.
  We denote by $\occ$ the full subcategory of $\pn$ with occurrence
  p-nets as objects.
\end{defi}

Observe that by the requiring all 
events secured
(together with the fact that the initial marking is made of the conditions
with empty pre-set) an occurrence p-net could only be cyclic for the presence of
back-pointers to persistent conditions. This latter possibility is
excluded structurally by the irredundancy assumption (Definition~\ref{de:well-formed}).
The connectedness requirement for the pre-set of 
conditions is trivially
satisfied for non-persistent 
conditions, since they can have at most one
event in their pre-set by condition (3). Instead, for persistent conditions
it is an actual constraint: it imposes that we cannot split the generators of
the condition
in two subsets not connected by consistency. Again the idea
is that, if this were the case, the condition
should be split into two
occurrences each having an element of the partition as pre-set.
This is reminiscent of the local connectedness requirement for
{\evstr} and, indeed, it will ensure that the {\evstr} extracted from
an occurrence p-net is locally connected.

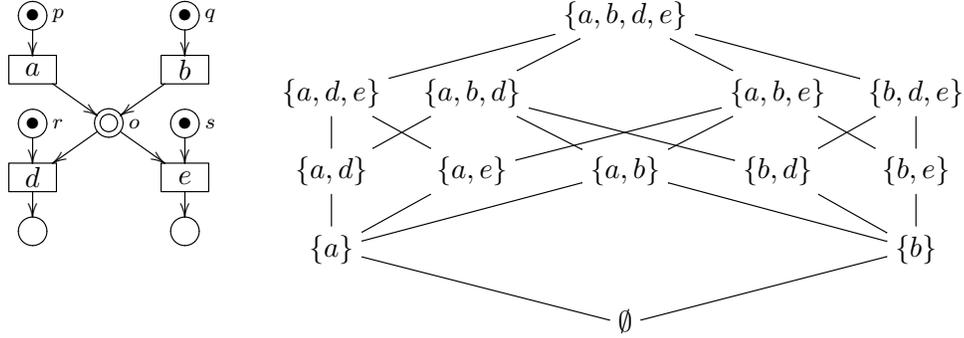
\begin{figure}[t]
$$
 \xymatrix@R=.8pc@C=1.2pc{
 &
 \drawmarkedplace\ar[d]
 \nameplaceright p
 & &
 \drawmarkedplace\ar[d]
 \nameplaceright q
 \\
 &
 \drawtrans a \ar[dr] 
 & &
 \drawtrans b \ar[dl] 
 \\
 &
 \drawmarkedplace\ar[d]
 \nameplaceright r
 &
 \drawpersistentplace\ar[dl] \ar[dr]
 \nameplaceright o
 &
 \drawmarkedplace\ar[d] 
 \nameplaceright s
 \\
 &
 \drawtrans d \ar[d]
 & &
 \drawtrans e \ar[d] 
 \\
 &
 \drawplace
 & &
 \drawplace
 }
 \qquad
 \xymatrix@R=1pc@C=.8pc{
 & & {\{a,b,d,e\}}
 \\
 {\{a,d,e\}} \ar@{-}[rru]  
 & {\{a,b,d\}} \ar@{-}[ru]
 & & {\{a,b,e\}} \ar@{-}[lu] 
 & {\{b,d,e\}} \ar@{-}[llu]
 \\
 {\{a,d\}} \ar@{-}[u] \ar@{-}[ru] 
 & {\{a,e\}} \ar@{-}[lu] \ar@{-}[rru]  
 & {\{a,b\}} \ar@{-}[lu] \ar@{-}[ru]  
 & {\{b,d\}} \ar@{-}[ru] \ar@{-}[llu]
 & {\{b,e\}} \ar@{-}[u] \ar@{-}[lu]
 \\
 {\{a\}} \ar@{-}[u] \ar@{-}[ru] \ar@{-}[rru]
 & & & & {\{b\}} \ar@{-}[u] \ar@{-}[lu] \ar@{-}[llu]
 \\
 & & {\emptyset} \ar@{-}[rru] \ar@{-}[llu]
 }
$$
\caption{An occurrence p-net %
  and its domain of configurations}
\label{fig:exoccpnet}
\end{figure}

An occurrence p-net inspired by our running example is in Figure~\ref{fig:exoccpnet} (left).
Event $a$ has a unique securing sequence consisting of $a$ itself, and similarly for $b$;
$d$ instead has two minimal securing sequences: $ad$ and $bd$. 
There is only one backward conflict, on the persistent condition
$o$, since 
$\post{a} \cap \post{b} = \{o\} \subseteq \pe{S}$. The conflict relation is empty and there is no causality between events.
Note in particular, that $\pre{o}$ is connected by $\conn{}_O$ since $a \conn{}_O b$.

Consider now, the p-net in Figure~\ref{fig:no-occurrence} (left). It
satisfies all conditions of Definition~\ref{def:occ-p-net} but (4). In
fact, now the pre-set of $\pre{o} = \{a,b\}$ is not connected by
$\conn{}_N$ since $a \#_N b$. Intuitively, to recover connectedness the place $o$
should be split in two places $o_1$ and $o_2$, with pre-sets $\{a\}$
and $\{b\}$, respectively, thus getting the occurrence p-net in
Figure~\ref{fig:no-occurrence} (right). Indeed, this is the occurrence p-net arising from the unfolding construction described below.

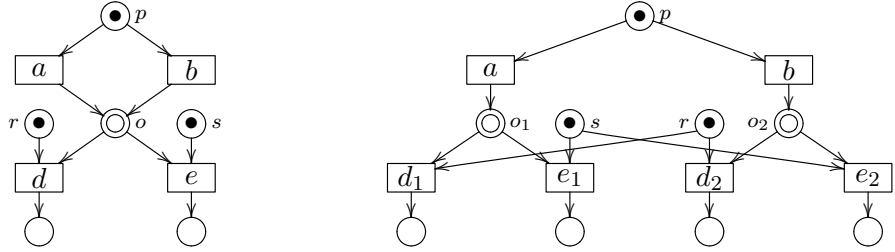
\begin{figure}[t]
$$
 \xymatrix@R=.8pc@C=1.2pc{
 &
 &
 \drawmarkedplace\ar[dl]\ar[dr]
 \nameplaceright p
 & 
 \\
 &
 \drawtrans a \ar[dr] 
 & &
 \drawtrans b \ar[dl] 
 \\
 &
 \drawmarkedplace\ar[d]
 \nameplaceleft r
 &
 \drawpersistentplace\ar[dl] \ar[dr]
 \nameplaceright o
 &
 \drawmarkedplace\ar[d] 
 \nameplaceright s
 \\
 &
 \drawtrans d \ar[d]
 & &
 \drawtrans e \ar[d] 
 \\
 &
 \drawplace
 & &
 \drawplace
}
\hspace{22mm}
\xymatrix@R=.8pc@C=1pc{
 &
 & &
 \drawmarkedplace\ar[dll]\ar[drr]
 \nameplaceright p
 & &
 \\
 & 
 \drawtrans a \ar[d] 
 & & & &
 \drawtrans b \ar[d] 
 \\
 & 
 \drawpersistentplace \ar[dl]\ar[dr]
 \nameplaceright {\ o_1}
 &
 \drawmarkedplace \ar[d]\ar@(l,r)[drrrr]
 \nameplaceright {s}
 & &
 \drawmarkedplace \ar[d] \ar@(r,l)[dllll]
 \nameplaceleft {r}
 &
 \drawpersistentplace \ar[dl]\ar[dr]
 \nameplaceleft {o_2\ }
 &
 \\
 \drawtrans {d_1} \ar[d] 
 & &
 \drawtrans {e_1} \ar[d]
 & &
 \drawtrans {d_2} \ar[d] 
 & &
 \drawtrans {e_2} \ar[d]
 \\
 \drawplace
 & 
 &
 \drawplace
 & &
 \drawplace
 &
 &
 \drawplace
}
$$

\caption{A p-net that is not an occurrence one (left) and its unfolding (right).}
\label{fig:no-occurrence}
\end{figure}

Note that, as in the case of ordinary occurrence nets, the initial
marking is determined by the structure of the net, i.e., it consists
of the set of non-persistent conditions
with empty pre-set. Still, for the sake of clarity it will be indicated explicitly.
It can be easily seen that, in absence of persistent conditions,
Definition~\ref{def:occpnet} gives the ordinary notion of occurrence net. We have indeed that
\begin{itemize}
\item there is no self-conflict because each event
is secured;
\item there is no backward conflict because for all $e, e' \in E$, if $e \neq e'$ 
then
  $\post{e} \cap \post{e'} = \emptyset$.
\end{itemize}

Some proofs of results on occurrence p-nets will exploit induction on the following notion of depth, that generalises the one of ordinary occurrence nets.

\begin{defi}[depth]
  Let $O = (B, \pe{B}, E, \gamma_0, \gamma_1, v_0)$ be an occurrence
  p-net. We define the \emph{depth} of events and conditions as
  \begin{eqnarray*}
  \depth{e} & = & 1+ \max \{ \depth{b} \mid b \in \pre{e} \}\\
  \depth{b} & = & \left\{\begin{array}{lcl}
  0 & \quad & \mbox{if $b \in v_0$} \\
  \min \{ \depth{e} \mid e \in \pre{b}\} & \quad & \mbox{otherwise}
  \end{array}\right.
  \end{eqnarray*}
\end{defi}

We introduce a notion of configuration for occurrence p-nets that is intended to capture the concept of a (concurrent) computation.

\begin{defi}[configuration of occurrence p-nets]
  Let $O= (B, \pe{B}, E, \gamma_0, \gamma_1, v_0)$ be an occurrence
  p-net. A configuration is a set of events $C \subseteq E$
  such that each $e \in C$ admits a securing sequence in $C$ and 
  for all $e, e' \in C$, $\neg (e \#_N e')$.
  We denote by $\conf{O}$ the set of
  configurations of $O$.
\end{defi}

The configurations of the net in Figure~\ref{fig:exoccpnet} (left) occur on the right, ordered by inclusion.

We next prove that configurations can be interpreted as
representations of classes of firing sequences where the order of
independent firings is abstracted.

\begin{lem}[configurations are executions]
  \label{le:conf-exec}
  Let $O= (B, \pe{B}, E, \gamma_0, \gamma_1, v_0)$ be an occurrence
  p-net. Then the following statements are equivalent
  \begin{enumerate}
  \item  $C \in \conf{O}$ is a finite configuration;
    
  \item there exists a securing sequence of events $e_1, \ldots, e_n$
    such that $C = \{ e_1, \ldots, e_n \}$;
    
  \item there exists a firing sequence
    $v_0 \fire{e_1} v_1 \fire{e_2} \ldots \fire{e_n} v_n$ such that
    $C = \{ e_1, \ldots, e_n \}$ and 
    $v_n = (v_0 \cup \bigcup_{i=1}^n \post{e_i}) \ominus (\bigcup_{i=1}^n \pre{e_i})$.
  \end{enumerate}
\end{lem}

\begin{proof}
  ($1 \leftrightarrow 2$) Let $C \in \conf{O}$ be a finite
  configuration. We prove the result by induction on $|C|$. The base
  case $|C|=0$ is trivial. 
  If $|C|=n+1$ for some $n\in\nat$,  
  then for
  each $e \in C$ fix a single securing sequence of minimal length 
  $s_e = e_1, \ldots, e_k, e$ in $C$,
  which exists by definition of configuration, 
  and call $S=\{s_e \mid e\in C\}$
  the set of such securing sequences. Let $\hat{s}_{\hat{e}} \in S$
  be one of such 
  sequences of maximal length in $S$, with $\hat{e}$ as last event.
  Clearly $C' = C \setminus \{\hat{e}\}$ is again a configuration
  because its events are not in conflict (as they are in $C$) and for
  each $e \in C'$ the securing sequence $s_e \in S$ consists of events
  in $C'$ only, otherwise the minimality of $\hat{s}_{\hat{e}}$ as a
  securing sequence for $\hat{e}$ in $C$ would be violated.
  Then, by inductive hypothesis, there exists a securing sequence
  $e_1 \ldots e_n$ such that $C' =\{e_1, \ldots, e_n \}$. It is
  immediate to see that $e_1, \ldots, e_n, \hat{e}$ is the desired securing
  sequence for $C$.

  The converse implication is immediate, just observing that securing
  sequences are closed by prefix and conflict free.

  \medskip
  
  ($2 \leftrightarrow 3$) Assume that there is a securing sequence
  $e_1, \ldots, e_n$ such that $C = \{ e_1, \ldots, e_n \}$. We
  proceed by induction on $n$. The base case $n=0$ is trivial. For
  $n>0$, we know, by inductive hypothesis that there exists a firing
  sequence
  $v_0 \fire{e_1} v_1 \fire{e_2} \ldots \fire{e_{n-1}} v_{n-1}$ and
  $v_{n-1} = (v_0 \cup \bigcup_{i=1}^{n-1} \post{e_i}) \ominus
  (\bigcup_{i=1}^{n-1} \pre{e_i})$. By definition of securing
  sequence, $\{ e_1, \ldots, e_{n-1} \} \vdash e_n$, we know that for
  all $b \in \pre{e_n}$ either $b\in v_0$ or there exists $i < n$ such
  that $b \in \post{e_i}$. Moreover, for all $i< n$, if
  $b \in \pre{e_n} \cap \npe{B}$, certainly $b \not\in \pre{e_i}$, otherwise we would
  have $e_i \# e_n$. This allows us to deduce that
  $\pre{e_n} \subseteq v_n$. Hence the firing sequence can be extended
  by $v_{n-1} \fire{e_n} v_n$, where
  $v_n = (v_{n-1} \cup \post{e_n}) \ominus \pre{e_n} = (v_0 \cup
  \bigcup_{i=1}^n \post{e_i}) \ominus (\bigcup_{i=1}^n \pre{e_i})$.

  For the converse implication, let us start showing by induction on
  $n$ that if $C = \{e_1, \ldots, e_n\}$ is the set of events of a
  given firing sequence, then
\begin{equation}
\mbox{there cannot 
be $i, j \in \interval{n}$ such that $i \not = j$ and $\pre{e_i} \cap \pre{e_j} \cap \npe{B} \not = \emptyset$.}
\label{eq:dagger}
\end{equation}
For  $n\leq 1$ the statement is trivial. Let $n>1$ and assume, by contradiction, that
  there is a non-persistent place  $b \in \pre{e_n} \cap \pre{e_i} \cap \npe{B}$, for $i < n$.
  Clearly $b \not \in v_0$, otherwise 
  $\pre{b} = \emptyset$ and after the firing of $e_i$  there would be no way of 
  generating the token in $b$. Therefore
  $\pre{b} \neq \emptyset$. Let $e \in \pre{b}$ be the only event
  generating a token in $b$. Then necessarily there are $j < i$ and
  $k < n$ such that $e_j = e_k = e$ and thus, since $O$ is
  t-restricted,
  $\pre{e_j} \cap \pre{e_k} \cap \npe{B} \not = \emptyset$, which is absurd because we assumed that the statement holds for firing sequences
  shorter than $n$.
  
  Now,  by Definition~\ref{de:net-securing-disjunct} we have to show that (1) all events of $C$ are 
distinct,  (2) $\neg \# C$, and  (3) for all $i \in \interval{n}$,
$\{e_1 \ldots e_{i-1}\} \vdash e_{i}$. Point (1) follows from~(\ref{eq:dagger}) and t-restrictedness. 
For (3), suppose by absurd that  $\{e_1 \ldots e_{i-1}\} \not\vdash e_{i}$ for some $i \in \interval{n}$.
Thus there is a place $b \in \pre{e_i}$ such that $b \not \in v_0$ and $b \not \in \post{e_j}$ for all $j \in
\interval{i-1}$, but this is impossible because by assumption $e_i$ is enabled in $v_{i-1}$. 

For (2) we proceed by induction on $n$, the base case $n=0$ being obvious. Suppose
by absurd that $\#C$ holds, and let $k \in \interval{n}$ be the minimal index such that 
$\#\{e_1, \ldots, e_k\}$. By this assumption we know that $\neg \#\{e_1,\ldots, e_{k-1}\}$,  
thus by induction hypothesis that $\{e_1, \ldots, e_{k-1}\}$ is a securing sequence, and by the 
implication ($2 \to 3$) already proved that 
    $v_0 \fire{e_1} v_1 \fire{e_2} \ldots \fire{e_{k-1}} v_{k-1}$ is a firing sequence and     
    $v_{k-1} = (v_0 \cup \bigcup_{i=1}^{k-1} \post{e_i}) \ominus (\bigcup_{i=1}^{k-1} \pre{e_i})$.

Now, since $\#\{e_1, \ldots, e_k\}$, by the clauses defining conflict in Definition~\ref{def:dependence} 
either (clause (a)) there is an event $e_h$ with $h \in \interval{k-1}$ such that $\pre{e_h} \cap \pre{e_k} 
\not \subseteq \pe{B}$, but this is impossible by~(\ref{eq:dagger}) above, or (clause (c)) w.l.o.g.~there 
is a place $b \in \pre{e_k}$ such that $\#\{e_1, \ldots, e_{k-1}, b\}$. Since $e_k$ is enabled in $v_{k-1}$,
certainly $b \in  v_{k-1} =  (v_0 \cup \bigcup_{i=1}^{k-1} \post{e_i}) \ominus (\bigcup_{i=1}^{k-1} \pre{e_i})$, and as it cannot belong to the initial marking, $b \in \post{e_j}$ for some $j \in \interval{k-1}$. But by (clause (b)) from $\#\{e_1, \ldots, e_{k-1}, b\}$ we can infer $\#\{e_1, \ldots, e_{k-1}, e_j\}$ and thus $\#\{e_1, \ldots, e_{k-1}\}$  because  $j \in \interval{k-1}$, which is absurd by the choice of $k$.
\end{proof}

By the above result, it is meaningful to define the marking reached after a configuration  for occurrence p-nets.

\begin{defi}[marking after a configuration]
  Let $O$ be an occurrence p-net. Given $C \in \conf{O}$, the marking after $C$ is $\Mark{C} = (v_0 \cup \bigcup_{e \in C} \post{e})\ominus (\bigcup_{e \in C} \pre{e})$.
\end{defi}

We next observe that all occurrence p-nets are safe.

\begin{prop}[occurrence p-nets are safe]
\label{pr:distinct}
  Let $O$ be an occurrence p-net. Then, for each firing sequence
  $v_0 \fire{e_1} v_1 \fire{e_2} \ldots \fire{e_n} v_n$ all markings
  $v_i$ are sets.
\end{prop}

\begin{proof}
By Lemma~\ref{le:conf-exec} we know that $\{e_1, \ldots, e_i\}$ is a securing sequence, 
thus all the events in it are pairwise distinct.
  The fact that all markings are safe immediately follows 
  recalling that for non-persistent places
  $|\pre{b}| \leq 1$, and for persistent places idempotency ensures
  that at most one token is in the place.
\end{proof}

We next introduce a notion of concurrency for
occurrence p-nets that, as anticipated, 
is based on non-binary conflict .

\begin{defi}[concurrency]
  \label{de:concurrency}
  Let %
  $O$ be an occurrence p-net. A subset of conditions $X \subseteq B$
  is called \emph{concurrent}, written $\conc{X}$, if $\neg \# X$ and
  for all $b, b' \in X$ if $b<b'$ then $b \in \pe{B}$.
\end{defi}

Concurrency, as in the case of ordinary nets, is intended to provide a
structural characterisation of coverability. This is formalised below in Lemma~\ref{le:cover-conc-general}.
Quite intuitively, a concurrent set of conditions cannot include
conflicts.
Note, instead, that causal dependencies from persistent places
are admitted, consistently with the fact that using a token in a
persistent place does not consume such token.

\begin{lem}[coverability vs concurrency]
  \label{le:cover-conc-general}
  Let $O$ 
  be an occurrence p-net
  and $X \subseteq B$. Then $X$ is
  concurrent iff there is a reachable marking that covers $X$.
\end{lem}

\begin{proof}
  $(\Rightarrow)$ Let $\conc{X}$.
  In order to prove that $X$ is
  coverable we proceed by induction on the pairs $h_X = \langle \mxdepth{X}, \cardmaxdepth{X}\rangle$
  where $\mxdepth{X} = max \{\depth{b}\mid {b \in X}\}$, and $\cardmaxdepth{X} =  
  |\{b \in X \mid \depth{b} = \mxdepth{X}\}|$, ordered by $\langle n,m\rangle < \langle n',m'\rangle$ if 
  $n < n'$ or $n = n' \wedge m < m'$. 
  
 If $h_X=\langle 0, \_\rangle$ then $X \subseteq v_0$,
  hence we conclude immediately. If $\mxdepth{X} > 0$, take a condition
  $b \in X$ such that $\depth{b}$ is maximal. Note that there must be
  $e \in \pre{b}$ such that $Y = (X \setminus \post{e}) \cup \pre{e}$
  is still concurrent, otherwise either $b$ would not be of maximal
  depth or, by the rules defining conflict, $\# X$.
  Clearly $h_Y < h_X$ hence, by
  inductive hypothesis, $Y$ is coverable and we can conclude that
  $X$ is coverable.

  \medskip
  
  $(\Leftarrow)$ We show that any reachable marking is
  concurrent. Since any subset of a concurrent set is concurrent this
  allows us immediately to conclude. Consider a firing sequence
  $v_0 \fire{e_1} v_1 \fire{e_2} \ldots \fire{e_n} v_n$. We show by
  induction on $n$ that $\conc{v_n}$. For the base case we just need
  to observe that $\conc{v_0}$, i.e., the initial marking is clearly
  concurrent. When $n >0$, by inductive hypothesis we know that
  $\conc{v_{n-1}}$. This allows us to deduce that also
  $v_n = (v_{n-1} \ominus \pre{e_n}) \oplus \post{e_n}$ is concurrent.
\end{proof}

For instance, consider the occurrence p-net in
Figure~\hyperref[fig:combined]{\ref*{fig:combined}a}. We already observed that
 $\# \{p,q,r\}$ and indeed such set is not coverable. Instead, each
pair of conditions in $\{p,q,r\}$ is concurrent and thus coverable. Interestingly enough, this shows that differently from what happens for ordinary occurrence nets, pairwise coverability does not imply coverability.

We can now show that morphisms of occurrence p-nets preserve
concurrency.

\begin{lem}[morphisms preserve concurrency]
  \label{le:morph-conc}
  Let $f : O \to  O'$ be
  a morphism of occurrence p-nets and $X \subseteq B$. If $\conc{X}$ then $\conc{\pl{f}(X)}$.
\end{lem}

\begin{proof}
  Since morphisms are simulations (see Lemma~\ref{lemma:simulation})
  they preserve coverability. By
  Lemma~\ref{le:cover-conc-general}, coverability is the same as
  concurrency. Hence $\conc{X}$ implies $\conc{\pl{f}(X)}$.
\end{proof}

\subsection{Occurrence p-nets with equivalence}

We now introduce the notion of occurrence p-net with equivalence,
which will be the target of the pre-unfolding construction. The
intuition is that, in an occurrence p-net with equivalence,
occurrences of items that depend on different disjunctive causes are
kept separate, but related by the equivalence. This is technically
useful in the development of the unfolding construction.

Let $A$ be a set and  ${\sim} \subseteq A \times A$ an equivalence relation on $A$.
Given $x\in A$ we write $\eqclass[\sim]{x} = \{ y\in A \mid x \sim y\}$ for the equivalence class of $x$.
Moreover, given $X\subseteq A$ we write $\eqclass[\sim]{X}$ for the set $\{ \eqclass[\sim]{x} \mid x \in X \}$.

Given two subsets $X,Y \subseteq A$ we write $X \sim Y$ instead of $\eqclass[\sim]{X} = \eqclass[\sim]{Y}$.
Should the equivalence be used (as it will  always be the case) for sets 
$X$ and $Y$
each consisting of pairwise non-equivalent elements,
from $X \sim Y$ it follows that there is a bijection from $X$ to $Y$
mapping each element $x \in X$ to the only $y \in Y$ such that
$x \sim y$.
If moreover $A$ is partially ordered by $\leq$, we write  $X \leq Y$ if for all $x \in X$, $y \in Y$ we have
$x \leq y$. 
The notation will be used only for $|X| \leq 1$. Observe
that $\emptyset \leq Y$ trivially holds, while $\{ x \} \leq Y$
reduces to $x \leq y$ for all $y \in Y$.

\begin{defi}[occurrence p-net with equivalence]
  \label{de:occ-eq}
  An occurrence p-net with equivalence is a pair
  $\langle O, \sim \rangle$ where
  $O=(B,\pe{B},E,\gamma_0,\gamma_1,v_0)$ is an occurrence p-net
  without backward conflicts and
  ${\sim} \subseteq (B \times B) \cup (E \times E)$ is an equivalence
  such that
  \begin{enumerate}
  \item for all $b, b' \in B$ with $b \sim b'$
    \begin{enumerate}      
    \item
      \label{de:occ-eq:1a}
      either $b, b' \in \pe{B}$ or $b, b' \in \npe{B}$;
      
    \item
      \label{de:occ-eq:1b}
      if $b \neq b'$ then $\neg (\pre{b} \leq \pre{b'})$
      and 
      $\{ b, b' \} \not\subseteq \pre{e}$
      for all $e \in E$;      

    \item
      \label{de:occ-eq:1c}
      if $b, b' \in \npe{B}$ then $\pre{b} \sim \pre{b'}$;

    \item
      \label{de:occ-eq:1d}
      if $b, b' \in \pe{B}$ then either $\pre{b} \sim \pre{b'}$ or $b\, (\sim \setminus~ \#)^* \, b'$;
    \end{enumerate}
    
  \item
    \label{de:occ-eq:2}
    for all $e, e' \in E$, if $e \sim e'$ and $e \neq e'$ then
    $\pre{e} \sim \pre{e'}$, $\pre{e} \neq \pre{e'}$, 
    and 
    $\post{e} \sim \post{e'}$;
    
  \item
    \label{de:occ-eq:3}
    for all $X, X' \subseteq B$, if $X \sim X'$ then $\{ e \mid \pre{e} = X\} \sim \{ e' \mid \pre{e'} = X' \}$.  
  \end{enumerate}
\end{defi}

In words, an occurrence p-net with equivalence is an occurrence p-net
where the absence of backward conflicts implies that each item has a uniquely determined causal history. 
The possibility of joining the different histories of tokens in
persistent places is captured by an equivalence that can equate
persistent places. Events are equated when they have equivalent
pre-sets and, in turn, this implies that also their post-sets are
equivalent.

More precisely, by condition~(\ref{de:occ-eq:1a}) equivalence respects
the sort of places: two equivalent places are either both
non-persistent or both persistent.
By the second part of
condition~(\ref{de:occ-eq:1b}) the pre-set of each event consists of
pairwise non-equivalent places. The first part is slightly more
complex. First note that whenever $b \in v_0$, i.e.
$\pre{b} = \emptyset$, the inequality $\pre{b} \leq \pre{b'}$ is
trivially satisfied. Thus, places in the initial marking are not
equivalent to any other place and, in particular, the initial marking
consists of pairwise non-equivalent places. If instead
$\pre{b} = \{e\}$, then $b$ cannot be equated to any place $b'$ caused
by $e$. In particular, this implies that the places in the post-set of
each event are pairwise non-equivalent. More generally, this condition
plays a role in ensuring that the quotiented p-net is irredundant
(Definition~\ref{de:well-formed}).
By condition~(\ref{de:occ-eq:1c}) non-persistent places can be
equivalent only when they are generated by equivalent events. Finally,
condition~(\ref{de:occ-eq:1d}) states that persistent places can be
equivalent even if they are not generated by equivalent events, but in
this case they must be connected by a chain of consistency in the
equivalence class. This ensures that, once quotiented, the pre-set of
the condition will satisfy the connectedness condition (see
Definition~\ref{def:occpnet}).

Condition~(\ref{de:occ-eq:2}) says that events are equivalent only when they have
equivalent pre-sets and in this case they also have equivalent
post-sets. Moreover, equivalent events must differ in their pre-set.

Finally, by condition~(\ref{de:occ-eq:3}), whenever two sets of
conditions $X, X'$ are equivalent then the sets of events having $X$
and $X'$ as pre-sets are equivalent. Note that since $\pre{e} = X$, by
condition (1b), the set $X$ cannot contain equivalent conditions. The
same applies to $X'$ and thus $X \sim X'$ implies that there is a
bijection from $X$ to $X'$ mapping each element $b_1 \in X$ to the
only $b_2 \in X'$ such that $b_1 \sim b_2$.  Moreover, also the set
$\{ e \mid \pre{e} = X \}$ contains pairwise non-equivalent events by
condition (2). The same applies to $\{ e' \mid \pre{e'} = X' \}$,
hence, also in this case, the requirement
$\{ e \mid \pre{e} = X \} \sim \{ e' \mid \pre{e'} = X' \}$ implies a
one-to-one correspondence between equivalent events.

An example of occurrence p-net with equivalence is in Figure~\ref{fig:occpneteq}, where equivalent elements are linked by dotted lines.

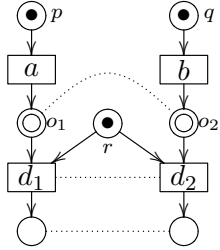
\begin{figure}[t]
$$
 \xymatrix@R=.8pc@C=1.2pc{
 \drawmarkedplace\ar[d]
 \nameplaceright p
 & &
 \drawmarkedplace\ar[d]
 \nameplaceright q
 \\
 \drawtrans a \ar[d] 
 & &
 \drawtrans b \ar[d] 
 \\
 \drawpersistentplace\ar[d] \ar@{.}@(ur,ul)[rr]
 \nameplaceright {o_1}
 &
 \drawmarkedplace\ar[dl]\ar[dr]
 \nameplacedown r
 &
 \drawpersistentplace\ar[d]
 \nameplaceright {o_2}
 \\
 \drawtrans {d_1} \ar[d]\ar@{.}[rr]
 & &
 \drawtrans {d_2} \ar[d] 
 \\
 \drawplace \ar@{.}[rr]
 & &
 \drawplace
 }
$$
\caption{An occurrence p-net with equivalence}\label{fig:occpneteq}
\end{figure}

The fact that occurrence p-nets with equivalence do not have backward conflicts allows us to restrict only to binary conflicts. Formally, the following holds.

\begin{lem}[conflicts in occurrence p-nets with equivalence]
  Let $\langle O, \sim \rangle$ be an occurrence p-net with
  equivalence and $X \subseteq B \cup E$. Then $\# X$ holds if and
  only if there are $x, x' \in X$ such that $x \# x'$.
\end{lem}

\begin{proof}
  Straightforward using the definition of conflict (see
  Definition~\ref{def:mainrels}).
\end{proof}

For occurrence p-nets with equivalence we will need a notion of concurrency on sets of places stronger than that in Definition~\ref{de:concurrency}.

\begin{defi}[strong concurrency]
  \label{def:sco}
  Let $\langle O, \sim \rangle$ be an occurrence p-net with
  equivalence. A subset of conditions $X \subseteq B$ is \emph{strongly
    concurrent}, written $\sconc{X}$, if $\conc{X}$ and for all
  $b, b' \in X$, if $b \neq b'$ then $\neg (b \sim b')$.
\end{defi}

The idea is that if a set of conditions is strongly concurrent it can be produced by a computation using only a specific instance for each equivalence class of persistent resources.

\begin{lem}[pre- and post-sets are strongly concurrent]
  \label{le:pre-post-strong}
  Let $\langle O, \sim \rangle$ be an occurrence p-net with
  equivalence. Then for any $e \in E$ its pre-set $\pre{e}$ and
  post-set $\post{e}$ are strongly concurrent.
\end{lem}

\begin{proof}
  Let $e \in E$ be an event of $O$. Then $\pre{e}$ and $\post{e}$ are
  concurrent since $O$ is an occurrence p-net, thus $e$ can be fired and
  hence pre- and post-sets of events are coverable, whence concurrent
  by Lemma~\ref{le:cover-conc-general}. They are also strongly concurrent,
  since, by Definition~\ref{de:occ-eq}, condition (1b), they cannot
  contain equivalent conditions.
\end{proof}

Occurrence p-nets with equivalence can be turned into a category by
introducing a suitable notion of morphism.

\begin{defi}[category of occurrence p-nets with equivalence]
  \label{de:occ-eq-morphism}
  A morphism of occurrence p-nets with equivalence
  $f : \langle O, \sim \rangle \to \langle O', \sim' \rangle$ is a
  morphism
  $f : O \to O'$
  of p-nets such that for all
  $b_1, b_2 \in B$,
  if $b_1 \sim b_2$ then  $\pl{f}(b_1) \sim' \pl{f}(b_2)$.
  We denote by $\occeq$ the category of occurrence p-nets with equivalence.  
\end{defi}

In words, a morphism of occurrence p-nets is required to preserve the equivalence.  This is essential to ensure that it induces a 
function on the quotient nets.

\begin{defi}[quotient of an occurrence p-net with equivalence]
  Let $\langle O, \sim \rangle$ be an occurrence p-net with
  equivalence. Its \emph{quotient} is the structure
  $$\occquot{\langle O, \sim \rangle} = (\eqclass[\sim]{B},
  \eqclass[\sim]{\pe{B}}, \eqclass[\sim]{E}, \gamma_0^\sim,
  \gamma_1^\sim, \eqclass[\sim]{v_0})$$ 
  where
  $\gamma_j^\sim(\eqclass[\sim]{e}) = \eqclass[\sim]{\gamma_j(e)}$ for
  $j \in \{0,1\}$.
\end{defi}

The quotient of the net in Figure~\ref{fig:occpneteq} is depicted in
Figure~\ref{fig:quotoccpneteq}.  Note that the quotient has introduced a
backward conflict on the place $o$.

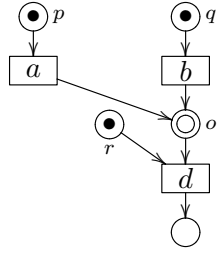
\begin{figure}[t]
$$
 \xymatrix@R=.8pc@C=1.2pc{
 \drawmarkedplace\ar[d]
 \nameplaceright p
 & &
 \drawmarkedplace\ar[d]
 \nameplaceright q
 \\
 \drawtrans a \ar[drr] 
 & &
 \drawtrans b \ar[d] 
 \\
 &
 \drawmarkedplace\ar[dr]
 \nameplacedown r
 &
 \drawpersistentplace\ar[d]
 \nameplaceright {o}
 \\
 & &
 \drawtrans {d} \ar[d] 
 \\
 & &
 \drawplace
 }
$$
\caption{The quotient of the occurrence p-net with equivalence in Figure~\ref{fig:occpneteq}}\label{fig:quotoccpneteq}
\end{figure}

We want to show that the quotient of an occurrence p-net with
equivalence is indeed an occurrence p-net. To this aim we first
observe some facts.

\begin{lem}[properties of the quotient]
  \label{le:q-properties}
  Let $\langle O, \sim \rangle$ be an occurrence p-net with
  equivalence and $\occquot{\langle O, \sim \rangle}$ the
  corresponding quotient.
  Then
  \begin{enumerate}
  \item
    \label{le:q-properties:1}
    for all $b \in \npe{B}$,
    $\pre{\eqclass[\sim]{b}} = \eqclass[\sim]{\pre{b}}$;
    
  \item
    \label{le:q-properties:2}
    for all $b \in \pe{B}$,
    $\pre{\eqclass[\sim]{b}} = \eqclass[\sim]{\bigcup_{b' \sim b} \pre{b'}}$;

  \item
    \label{le:q-properties:3}
    for all $e \in E$,
    $\pre{\eqclass[\sim]{e}} = \eqclass[\sim]{\pre{e}}$ and
    $\post{\eqclass[\sim]{e}} = \eqclass[\sim]{\post{e}}$;
    
  \item
    \label{le:q-properties:4}
    for all $X \subseteq E$ and $e \in E$, if $X \vdash e$ then
    $\eqclass[\sim]{X} \vdash \eqclass[\sim]{e}$;

  \item
    \label{le:q-properties:5}
    for all $e, e' \in E$, if $e \sim e'$ and $e \not = e'$  then $e \# e'$;
    
  \item
    \label{le:q-properties:6}
    for all $e, e' \in E$, if
    $\eqclass[\sim]{e} \# \eqclass[\sim]{e'}$ then $e \# e'$.

  \end{enumerate}
\end{lem}

\begin{proof}
  The proofs are mostly routine.

  For point~(\ref{le:q-properties:1}) observe that, by
  Definition~\ref{de:occ-eq}(\ref{de:occ-eq:1c}), if
  $b, b' \in \npe{B}$ then $\pre{b} \sim
  \pre{b'}$. Point~(\ref{le:q-properties:2}) is just a consequence of
  the definition of quotient. For point~(\ref{le:q-properties:3})
  observe that, by definition, if $e \sim e'$ then
  $\pre{e} \sim \pre{e'}$ and $\post{e} \sim \post{e'}$.

  Point~(\ref{le:q-properties:4}) follows from the observation that,
  by (3), $\pre{\eqclass[\sim]{e}} =
  \eqclass[\sim]{\pre{e}}$. Moreover, for all $b \in \pre{e}$, if
  $b \in v_0$ then $\eqclass[\sim]{b} \in
  \eqclass[\sim]{v_0}$. Otherwise, there is $e' \in X$ such that
  $b \in \post{e'}$ and thus, by~~(\ref{le:q-properties:3}),
  $\eqclass[\sim]{b} \in \post{\eqclass[\sim]{e'}}$.

  Concerning point~(\ref{le:q-properties:5}), let $e \sim e'$. Since
  in a p-net $\pre{e'} \not\subseteq \pe{B}$, we can fix some
  $b \in \pre{e} \cap \npe{B}$. By definition of occurrence p-net with
  equivalence (Definition~\ref{de:occ-eq}(\ref{de:occ-eq:1c})),
  $\pre{e} \sim \pre{e'}$, hence there must be
  $b' \in \pre{e'} \cap \npe{B}$ such that $b \sim b'$.
  We proceed by induction on the depth $h$ of $e$.
  When $h = 0$, we have $b \in v_0$. Again, by
  Definition~\ref{de:occ-eq}(\ref{de:occ-eq:1c}),
  $\pre{b} \sim \pre{b'}$, and thus also $b' \in v_0$. Since
  $b \sim b'$ we deduce that $b =b'$ and thus $e \# e'$. If
  $h >0$, there are $e_1 \in \pre{b}$ and $e_1' \in \pre{b'}$ and,
  since $b, b' \in \npe{B}$, it holds $e_1 \sim e_1'$. Hence by
  induction hypothesis $e_1 \# e_1'$ and thus $e \# e'$.

  Finally, for point~(\ref{le:q-properties:6}), we can proceed by
  induction on the derivation on rules that define conflict (see
  Definition~\ref{def:mainrels}). Let
  $\eqclass[\sim]{e} \# \eqclass[\sim]{e'}$.  If there is
  $\eqclass[\sim]{b} \in \pre{\eqclass[\sim]{e}} \cap
  \pre{\eqclass[\sim]{e'}} \cap \eqclass[\sim]{\npe{B}}$ then we
  conclude by point (3) that $e \# e'$.
  If instead there are
  $\eqclass[\sim]{b} \in \pre{\eqclass[\sim]{e}}=
  \eqclass[\sim]{\pre{e}}$ and
  $\eqclass[\sim]{b'} \in \pre{\eqclass[\sim]{e'}}=
  \eqclass[\sim]{\pre{e'}}$ with
  $\eqclass[\sim]{b}\# \eqclass[\sim]{b'} $, $b\in\pre{e}$ and
  $b'\in \pre{e'}$ it means that for any
  $\eqclass[\sim]{e_1} \in \pre{\eqclass[\sim]{b}}$ and any
  $\eqclass[\sim]{e_2} \in \pre{\eqclass[\sim]{b'}}$ we have
  $\eqclass[\sim]{e_1}\# \eqclass[\sim]{e_2}$. In particular, for
  $e_1 \in \pre{b}$ and $e_2\in \pre{b'}$ (that are uniquely
  determined as the underlying occurrence p-net as no backward
  conflict) by inductive hypothesis we have $e_1 \# e_2$ and thus
  $b \# b'$ and $e\# e'$.
\end{proof}

We can now reach the desired conclusion.

\begin{lem}[quotient is an occurrence p-net]
  \label{le:quotient-well-defined}
  Let $\langle O, \sim \rangle$ be an occurrence p-net with
  equivalence. Then $\occquot{\langle O, \sim \rangle}$ is an occurrence p-net.
\end{lem}

\begin{proof}
  Easy consequence of Lemma~\ref{le:q-properties}.
  In particular, for any event $e \in E$, by
  Lemma~\ref{le:q-properties}(\ref{le:q-properties:3}) we have that
  $\pre{\eqclass[\sim]{e}} = \eqclass[\sim]{\pre{e}}$. Then
  connectedness follows from
  Lemma~\ref{le:q-properties}(\ref{le:q-properties:6}). In order to
  conclude we have to observe that $\occquot{\langle O, \sim \rangle}$
  is a well-formed p-net. T-restrictedness follows immediately from
  t-restrictedness of $O$. Concerning irredundancy, consider a generic
  persistent place in $\occquot{\langle O, \sim \rangle}$ that will be
  of the kind $\eqclass[\sim]{b}$ for $b \in \pe{B}$ and take any
  event $\eqclass[\sim]{e} \in \pre{\eqclass[\sim]{b}}$. By
  Lemma~\ref{le:q-properties}(\ref{le:q-properties:2}),
  $\pre{\eqclass[\sim]{b}} = \eqclass[\sim]{\bigcup_{b' \sim b}
    \pre{b'}}$, hence we can assume $e \in \pre{b}$. In order to
  violate irredundancy, there should exist $e' \sim e$ and $b' \sim b$
  such that $e' \struct_O^n b'$ with $n \geq 2$. It is easy to see
  that this cannot happen thanks to conditions~(\ref{de:occ-eq:1b})
  and~(\ref{de:occ-eq:2}) of Definition~\ref{de:occ-eq}.
\end{proof}

We can thus consider a quotient functor from the category of p-nets with equivalence to the category of occurrence p-nets.

\begin{defi}[quotient functor]
  We denote by $\occquot{} : \occeq \to \occ$ the functor taking an
  occurrence p-net with equivalence to its quotient. 
\end{defi}

The fact that the functor is well-defined on objects follows
from Lemma~\ref{le:quotient-well-defined}. 
On arrows it is immediate from the definition of morphism of occurrence net with equivalence (Definition~\ref{de:occ-eq-morphism}).

\begin{lem}[occurrence p-net with equivalence vs occurrence p-nets]
  \label{le:occ-eq-occ}
  Let $\langle O, \sim \rangle$ 
  be an occurrence p-net with
  equivalence. For any finite configuration $C \in \conf{O}$,
  $\eqclass[\sim]{C} \in \conf{\occquot{O}}$ and
  $\Mark{\eqclass[\sim]{C}} = \eqclass[\sim]{\Mark{C}}$.
\end{lem}

\begin{proof}
  Immediate from Lemma~\ref{le:q-properties}, points
  (\ref{le:q-properties:4}) and (\ref{le:q-properties:6}).
\end{proof}

\subsection{Unfolding}

As mentioned above, the first phase of the unfolding construction produces an occurrence p-net with equivalence, which is then quotiented to an occurrence p-net. We will use $\pi$ to denote the projection on the first component of a pair, i.e., given sets $A$ and $B$, we let $\pi : A \times B \to A$ be defined as $\pi(a,b) = a$ for all $(a,b) \in A \times B$.

\begin{defi}[unfolding]
 \label{def:unf}
 Let $N=(S,\pe{S},T,\delta_0,\delta_1,u_0)$ be a p-net.  Define the
 \emph{pre-unfolding}
 $\preUnf{N} = (B,\pe{B},E, \gamma_0, \gamma_1,v_0)$ as the least
 occurrence p-net with an equivalence $\sim_N$ on $B \cup E$ such that
 \begin{itemize}
   
 \item $v_0 = \{ \langle s, \bot \rangle \mid s \in u_0 \} \subseteq B$
 \item if $t \in T$ and $X \subseteq B$ such that $\sconc{X}$ and
   $\pi(X) = \pre{t}$ then $e = \langle t, X \rangle \in E$,
   $Y = \{ \langle s, e \rangle \mid s \in \post{t} \} \subseteq B$ and
   $\pre{e} = X$, $\post{e} = Y$.
 \end{itemize}
 
 The
 equivalence $\sim_N$ is the least equivalence relation that satisfies
 \[
   \infer
   {\langle s, e\rangle \sim_N \langle s, e' \rangle}
   {s \in \pe{S}\ & e \conn{}_N e'
   }
   \quad
   \quad
   \infer
   {\langle t, X \rangle \sim_N \langle t, X' \rangle}
   {X \sim_N X'}
   \quad\quad
   \infer
   {\langle s, e \rangle \sim_N \langle s, e' \rangle}
   {e \sim_N e'}
 \]

 The unfolding $\unf{N}$ is obtained as the quotient
 $\quotient{\preUnf{N}}{\sim_N}$ with the \emph{folding morphism} $\varepsilon_N : \unf{N} \to N$
 defined as $\varepsilon_N(\eqclass[\sim_N]{x}) = \pi(x)$.
\end{defi}

Note that if the net $N$ does not contain persistent places, the equivalence relation $\sim_N$ is just the identity relation, strong concurrency coincide with concurrency and the pre-unfolding coincides with the ordinary unfolding.
Observe that $\varepsilon_N$ is injective on pre- and post-sets of
transitions, as implied by the definition of p-net morphism.

\begin{figure}[t]
$$
 \xymatrix@R=.8pc@C=1.2pc{
 & &
 \drawmarkedplace \ar@/_2.2pc/[ddddll] \ar@/^2.2pc/[ddddrr] 
 \nameplacedown s
 \\
 &
 \drawmarkedplace\ar[d]
 \nameplaceright p
 & &
 \drawmarkedplace\ar[d]
 \nameplaceright q
 \\
 &
 \drawtrans a \ar[d] 
 & &
 \drawtrans b \ar[d] 
 \\
 &
 \drawpersistentplace\ar[d] \ar[ld] \ar@/_1.2pc/[ddd] \ar@/_2pc/[ddddd] \ar@{.}@(ur,ul)[rr]
 \nameplaceright {o_1}
 &
 \drawmarkedplace\ar[d] \ar[ld] \ar[rd]
 \nameplaceright {r_1}
 &
 \drawpersistentplace\ar[d] \ar[rd] \ar@/^1.2pc/[ddd] \ar@/^2pc/[ddddd] 
 \nameplaceright {o_2}
 \\
 \drawtrans {e_1} \ar[d] \ar@{.}@(dr,dl)[rrrr]
 &
 \drawtrans {d_1} \ar[d] \ar@{.}@/_1pc/[rr]
 & 
 \drawtrans {c_1} \ar[d] 
 &
 \drawtrans {d_2} \ar[d] 
 &
 \drawtrans {e_2} \ar[d] 
 \\
 \drawplace \ar@{.}@(dr,dl)[rrrr]
 &
 \drawplace \ar@{.}@/_1pc/[rr]
 &
 \drawplace \ar[d] \ar[ld] \ar[rd]
 \nameplaceright {r_2}
 &
 \drawplace 
 &
 \drawplace
 \\
 &
 \drawtrans {d_3} \ar[d] \ar@{.}@/_1pc/[rr]
 & 
 \drawtrans {c_2} \ar[d] 
 &
 \drawtrans {d_4} \ar[d] 
 \\
 &
 \drawplace \ar@{.}@/_1pc/[rr]
 &
 \drawplace \ar[d] \ar[ld] \ar[rd]
 \nameplaceright {r_3}
 &
 \drawplace
 \\
 &
 \drawtrans {d_5} \ar[d] \ar@{.}@/_1pc/[rr]
 & 
 \drawtrans {c_3} \ar[d] 
 &
 \drawtrans {d_6} \ar[d] 
 \\
 &
 \drawplace  \ar@{.}@/_1pc/[rr]
 &
 \drawplace \ar@{--}[d]
 \nameplaceright {r_4}
 &
 \drawplace 
 \\
 &
 &
 &
 }
$$
\caption{Pre-unfolding of our running example}\label{fig:preunf}
\end{figure}
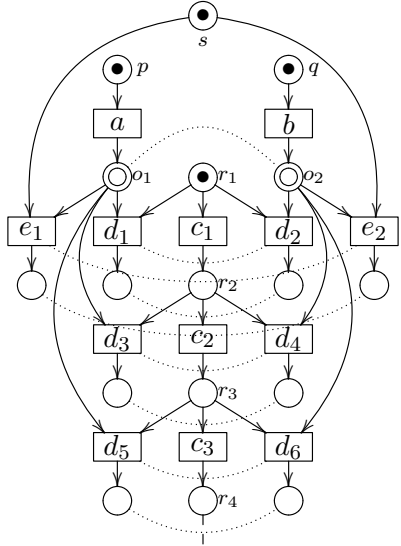

As an example, consider the (fragment of the) pre-unfolding of our running example in Figure~\ref{fig:preunf}, where we have used convenient names for places and transitions in order to improve readability.
Since $o_1$ and $o_2$ are instances of the same persistent place $o$, they are related by the equivalence.
Then $\{o_1,r_1\} \sim \{o_2,r_1\}$ and $\{o_1,s\} \sim \{o_2,s\}$, thus $d_1 \sim d_2$ and $e_1\sim e_2$ and the equivalence is propagated to their post-sets. The same pattern is iterated for any instance of $r$ created by the subsequent firing of (the instances of) $c$.
The corresponding unfolding is obtained as the quotient in Figure~\ref{fig:unf}.

For the p-net in Figure~\ref{fig:no-occurrence} (left), the
pre-unfolding is the occurrence p-net in Figure~\ref{fig:no-occurrence}
(right). Note that $o_1$ and $o_2$ are not equivalent since their
generating events are in conflict and are not connected by a chain of
consistency. 
Since $\sim_N$ here is the identity, the unfolding coincides with the pre-unfolding.

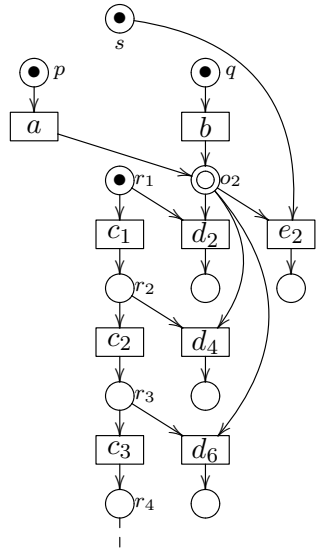
\begin{figure}[t]
$$
 \xymatrix@R=.8pc@C=1.2pc{
 & &
 \drawmarkedplace \ar@/^2.2pc/[ddddrr] 
 \nameplacedown s
 \\
 &
 \drawmarkedplace\ar[d]
 \nameplaceright p
 & &
 \drawmarkedplace\ar[d]
 \nameplaceright q
 \\
 &
 \drawtrans a \ar[drr] 
 & &
 \drawtrans b \ar[d] 
 \\
 &
 &
 \drawmarkedplace\ar[d] \ar[rd]
 \nameplaceright {r_1}
 &
 \drawpersistentplace\ar[d] \ar[rd] \ar@/^1.2pc/[ddd] \ar@/^2pc/[ddddd] 
 \nameplaceright {o_2}
 \\
 &
 & 
 \drawtrans {c_1} \ar[d] 
 &
 \drawtrans {d_2} \ar[d] 
 &
 \drawtrans {e_2} \ar[d] 
 \\
 &
 &
 \drawplace \ar[d] \ar[rd]
 \nameplaceright {r_2}
 &
 \drawplace 
 &
 \drawplace
 \\
 &
 & 
 \drawtrans {c_2} \ar[d] 
 &
 \drawtrans {d_4} \ar[d] 
 \\
 &
 &
 \drawplace \ar[d] \ar[rd]
 \nameplaceright {r_3}
 &
 \drawplace
 \\
 &
 & 
 \drawtrans {c_3} \ar[d] 
 &
 \drawtrans {d_6} \ar[d] 
 \\
 &
 &
 \drawplace \ar@{--}[d]
 \nameplaceright {r_4}
 &
 \drawplace 
 \\
 &
 &
 &
 }
$$
\caption{Unfolding of our running example}\label{fig:unf}
\end{figure}

We can consider the inclusion functor $\Ioc{} : \occ \to \pn$ that
acts as identity on objects and morphisms.  We next observe that the
unfolding $\unf{N}$ and the folding morphism $\varepsilon_N$ are cofree over
$N$.  Therefore $\unf{}$ extends to a functor that is right adjoint of
$\Ioc{}$ and establishes a coreflection between $\pn$ and
$\occ$.

It is easy to show that the pre-unfolding is an occurrence p-net with equivalence.

\begin{lem}[pre-unfolding is an occurrence p-net with equivalence]
  Let $N$
  be a p-net.  Then
  $\langle \preUnf{N}, \sim_N \rangle$ is an occurrence p-net with
  equivalence.
\end{lem}

\begin{proof}
  The fact that $\langle \preUnf{N}, \sim_N \rangle$ satisfies the
  properties in Definition~\ref{de:occ-eq} follows almost directly by
  construction. The less immediate property
  is~(\ref{de:occ-eq:1b}), specifically the fact that if $b \neq b'$
  then $\neg (\pre{b} \leq \pre{b'})$. This is a consequence of the
  fact that if $e \sim e'$ and $e \neq e'$ then $e \# e'$. 
\end{proof}

As a preliminary step we show that the pre-unfolding
construction $\preUnf{}$ extends to a functor that, together with the quotient functor $\occquot{}$, establishes an equivalence between the categories $\occ$ and $\occeq$.

\begin{lem}[mapping from the pre-unfolding]
  \label{le:map-pre-unf}
  Let $O$ be an occurrence p-net and
  $\pi : \preUnf{O} \to O$ the mapping from the pre-unfolding. Then
  \begin{enumerate}
  \item for all finite $C \in \conf{O}$ there exists
    $C' \in \conf{\preUnf{O}}$ such that $\pi(C') = C$ and
    $\Mark{C} = \pi(\Mark{C'})$;

  \item for all $C', C'' \in \conf{\preUnf{O}}$, if
    $\pi(C') = \pi(C'')$ then $C' \sim_N C''$.
  \end{enumerate}
\end{lem}

\begin{proof}
  Point (1) can be shown by induction on $|C|$. The base case $|C|=0$
  is trivial. If $|C| = n >0$, by Lemma~\ref{le:conf-exec}, we have
  $C = \{ e_1, \ldots, e_n \}$ and there is a firing sequence
  $v_0 \fire{e_1} v_1 \fire{e_2} \ldots \fire{e_n} v_{n}$. By the same
  lemma $C_1= \{ e_1, \ldots, e_{n-1} \}$ is a configuration of
  $O$. Hence by inductive hypothesis there exists a configuration
  $C_1' \in \conf{\preUnf{O}}$ such that $\pi(C_1') =C_1$ and
  $v_n = \Mark{C_1} = \pi(\Mark{C_1'})$ and $\conc{\Mark{C_1'}}$.
  Since $\pre{e_n} \subseteq v_{n-1}$, by definition of $\preUnf{}$
  there is an event $e_n' = \langle e_n, X \rangle$, where
  $X \subseteq v_{n-1}$ is such that $\sconc{X}$ and
  $\pi(X) = \pre{e_n}$. Note that we can always choose $X$ such that
  $\sconc{X}$ since equivalent conditions have the same
  $\pi$-image. If we define $C' = C_1 \cup \{e_n'\}$, we have that
  $\pi(C') =C$.

  Point (2) can be proved by an easy induction on $|C'|$.
\end{proof}

\begin{prop}[equivalence]
  \label{pr:occ-occeq-equiv}
  The categories $\occ$ and $\occeq$ are equivalent via the functors
  $\occquot{}$ and $\preUnf{}$.
\end{prop}

\begin{proof}
  Let us first observe that, on objects, the functors are inverse of
  each other. First, for an occurrence p-net $O$ let us
  define $q : \occquot{\preUnf{O}} \to O$ by letting
  $q(\eqclass[\sim]{x}) = \pi(x)$.

  In order to show that $q$ is an isomorphism we resort to
  Lemma~\ref{le:map-pre-unf}. By (1) we can show that $q$ is
  surjective. In fact, for each event $e$ in $O$ there is $e'$ in
  $\preUnf{O}$ such that $\pi(e') =e$. Hence
  $q(\eqclass[\sim]{e'}) = e$. Moreover, it is injective, since by (2)
  it follows that if $q(\eqclass[\sim]{x}) = q(\eqclass[\sim]{y})$
  then $x \sim y$.
  
  For the converse, define
  $\epsilon : \preUnf{\occquot{\langle O, \sim \rangle}} \to \langle
  O, \sim \rangle$ inductively as follows.

  The initial marking of $\preUnf{\occquot{O, \sim \rangle}}$ consists
  of the  places $b' =  \langle \eqclass[\sim]{b}, \bot  \rangle$ with
  $b  \in  v_0$.  We  define   $\epsilon(b')  =  b$.

  For the inductive step, consider an event
  $e' = \langle \eqclass[\sim]{e}, X' \rangle$ where
  $\eqclass[\sim]{e}$ is an event in
  $\occquot{\langle O, \sim \rangle}$ and $X' \subseteq B'$ such that
  $\sconc{X}$ and $\pi(X) = \pre{\eqclass[\sim]{e}}$. Since $X$ is
  strongly concurrent and $\epsilon$ preserves strong concurrency
  (Lemma~\ref{le:morph-conc}), we have that $\epsilon(X)$ is strong
  concurrent.  Moreover, by Lemma~\ref{le:q-properties}(3) we have
  $\pi(X) = \pre{\eqclass[\sim]{e}} = \eqclass[\sim]{\pre{e}}$.  Hence
  $\epsilon(X) \sim \pre{e}$ and therefore, by definition of
  occurrence p-net with equivalence, there exists a unique $e_1$ in $O$
  such that $e_1 \sim e$ and $\pre{e_1} =\epsilon(X)$. It exists by
  Definition~\ref{de:occ-eq}(\ref{de:occ-eq:3}) and it is unique by
  Definition~\ref{de:occ-eq}(\ref{de:occ-eq:2}). We define
  $\epsilon(e') = e_1$. Finally, $\epsilon$ is extended to $\post{e'}$
  by letting $\epsilon(\langle b, e' \rangle) = b_1$, where
  $b_1 \in \post{e_1}$ is the unique condition in $\post{e_1}$ such
  that $b_1 \sim b$. In this way $\epsilon(\post{e'}) = \post{e_1}$.

  It is easy to see that $\epsilon$ is injective. Surjectivity follows
  from Lemma~\ref{le:occ-eq-occ}.

  We finally show that for any occurrence p-net $O$, for any occurrence p-net
  with equivalence $\langle O', \sim \rangle$ and for any morphism
  $f :\occquot{\langle O', \sim \rangle} \to O$ there exists a unique
  morphism $h: \langle O', \sim \rangle \to \preUnf{O}$ such that the
  following diagram commutes
  \[
    \xymatrix{
      {\occquot{\preUnf{O}}} \ar[r]^(.55){\varepsilon_{O}} & {O} \\
      {\occquot{\langle O', \sim \rangle}} \ar@{.>}[u]^{\occquot{h}} \ar[ur]_{f}
    }
  \]

   We define $h$ by induction on the depth $k$ of the items.

   $\mathsf{(k=0)}$ Only conditions in the initial marking have
   depth $0$. Let $b' \in B'$ with $\depth{b'} = 0$. Hence
   $b' \in v_0'$ and thus $\eqclass[\sim]{b'}$ is in the initial
   marking of $\occquot{\langle O', \sim \rangle}$.  Therefore
   $\pl{f}(\eqclass[\sim]{b'}) = b \in v_0$ and we can define
   $\pl{h}(b') = \langle b, \bot \rangle$, which is in the initial
   marking of $\preUnf{O}$.

   \smallskip
  
   $\mathsf{(k >0)}$
   Let $e' \in E'$ be an event such that $\depth{e'} = k>0$. Therefore
   conditions in $\pre{e'}$ have depth less than $k$.  Hence their
   $h$-images have been already defined. Moreover, since $\pre{e'}$ is
   concurrent, by Lemma~\ref{le:morph-conc}, also $\pl{h}(\pre{e'})$
   is. Moreover, on pre-sets morphisms preserve equivalence, hence $h$
   is injective on $\pre{e'}$ and $\pl{h}(\pre{e'})$ does not include
   equivalent conditions, hence $\sconc{\pl{h}(\pre{e'})}$. Therefore the
   unfolding contains an event
   $e = \langle \tr{f}(e'), \pl{h}(\pre{e'}) \rangle$ and we define
   $\tr{h}(e') = e$. The mapping is then extended to the post-set of
   $e'$ by defining, for each $b \in\post{e'}$,
   $\pl{h}(b)= \langle \pl{f}(\eqclass[\sim]{b}), e\rangle$.

   Uniqueness follows by noticing that at each level we were forced to
   define $h$ as we did to ensure commutativity.
 \end{proof}

We can finally prove the desired theorem, by showing that $\preUnf{} : \pn \to \occeq$ is right adjoint to $\occquot{}: \occeq \to \pn$.

\begin{prop}
  \label{pr:cor-pre-unf}
  $\occquot{} \dashv \preUnf{}$
\end{prop}

\begin{proof}
  Let $N = (S, \pe{S}, T, \delta_0, \delta_1, u_0)$ be a p-net, let
  $\preUnf{N} = \langle O, \sim_N \rangle$ be
  its pre-unfolding with $O = \langle B, \pe{B}, E, \gamma_0, \gamma_1, v_0 \rangle$, and let $\varepsilon_N : \occquot{\preUnf{N}} \to N$ be the folding
  morphism. We have to show that for any occurrence p-net with equivalence $\langle O', \sim \rangle$ with $O' = \langle B', \pe{B'}, E', \gamma_0', \gamma_1', v_0' \rangle$ and for
  any morphism $f :\occquot{\langle O', \sim \rangle} \to N$ there exists a unique morphism
  $h: \langle O', \sim \rangle \to \preUnf{N}$ such that the following diagram commutes:
  \[
  \xymatrix{
    {\occquot{\preUnf{N}}} \ar[r]^(.55){\varepsilon_{N}} & {N} \\
    {\occquot{\langle O', \sim \rangle}} \ar@{.>}[u]^{\occquot{h}} \ar[ur]_{f}
    }
    \]

  We define $h$ by induction on the depth $k$ of the items.

  $\mathsf{(k=0)}$ Only conditions in the initial marking have
   depth $0$. Let $b' \in B'$ with $\depth{b'} = 0$. Hence
   $b' \in v_0'$ and thus $\eqclass[\sim]{b'} =\{b'\}$ is in the initial
   marking of $\occquot{\langle O', \sim \rangle}$.  Therefore
   $\pl{f}(\eqclass[\sim]{b'}) = b \in v_0$ and we can define
   $\pl{h}(b') = \langle b, \bot \rangle$, which is in the initial
   marking of $\preUnf{O}$.
  
   \smallskip
   
   $\mathsf{(k >0)}$
   Let $e' \in E'$ be an event such that $\depth{e'} = k>0$. Therefore
   conditions in $\pre{e'}$ have depth less than $k$.  Hence their
   $h$-images have been already defined. We have to define $\tr{h}$ on
   $e'$ and $\pl{h}$ on its post-set.
   \begin{itemize}
     
   \item
     If $\tr{h}(\eqclass[\sim]{e'})$ is undefined then necessarily
     also $\tr{h}(e')$ is undefined and $\pl{h}(b') = 0$ for all
     $b' \in \post{e'}$.

   \item
     If instead $\tr{h}(\eqclass[\sim]{e'}) = t$ observe that since
     $\pre{e'}$ is strongly concurrent by Lemma~\ref{le:morph-conc}
     also $\pl{h}(\pre{e'})$ is so. And, as argued in the proof of
     Proposition~\ref{pr:occ-occeq-equiv}, $\pl{h}(\pre{e'})$ is
     strongly concurrent.
     Now, since $f$ is a p-net morphism
     $\pl{f}(\pre{\eqclass[\sim]{e'}}) = \pre{t}$. Hence,  we have
     \begin{quote}
       $
       \begin{array}{ll}
         \pl{f}(\pre{\eqclass[\sim]{e'}}) = & \\
         \quad  = \pl{{\varepsilon_{N}}}(\pl{\occquot{h}}(\pre{\eqclass[\sim]{e'}})) & \mbox{[by
                                                                                       commutativity (up to depth $k$)]}\\
         \quad  = \pl{{\varepsilon_{N}}}(\eqclass[\sim]{\pl{h}(\pre{e'})}) & \mbox{[by definition of $\occquot{}$]}\\         
       \end{array}
       $
     \end{quote}
     
     Recalling the definition of $\varepsilon_N$, the above equality
     implies that
     $\pi(\pl{\occquot{h}}(\pre{e'})) = \pl{f}(\pre{\eqclass[\sim]{e'}})
     = \pre{t}$.
     Hence the unfolding contains an event
     $e = \langle t, \pl{h}(\pre{e'}) \rangle$ and we define $\tr{h}(e') =
     e$. The mapping is then extended to the post-set of
     $e'$ by defining, for each $b \in\post{e'}$, $\pl{h}(b)= \langle
     \pl{f}(\eqclass[\sim]{b}), e\rangle$.
   \end{itemize}

   Uniqueness follows by noticing that at each level we were forced to
   define $h$ as we did to ensure commutativity.
\end{proof}

\begin{cor}[coreflection between $\pn$ and $\occ$]
  \label{th:ws-to-cocc}
  $\Ioc{} \dashv \unf{}$
\end{cor}

\begin{proof}
  Immediate corollary of Propositions ~\ref{pr:occ-occeq-equiv}
  and~\ref{pr:cor-pre-unf}.
\end{proof}

\section{Locally Connected Event Structures from Occurrence p-Nets}
\label{sec:les}

In this section we show how an {\evstr} can be extracted from
an occurrence p-net, thus providing an {\evstr} semantics to
p-nets via the unfolding semantics. The transformation maps
occurrence p-nets to locally connected {\evstr}s and it is shown to be
functorial. Conversely, we show how a canonical occurrence p-net can
be associated to any locally connected {\evstr}. The two
transformations are shown to establish a coreflection.

An occurrence p-net can be easily mapped to an {\evstr} by forgetting
the conditions and keeping the events and the enabling and conflict
relations on events. The transformation gives rise to a functor from the 
category of occurrence p-nets to the category of locally connected {\evstr}.

\begin{defi}[{\evstr} for an occurrence p-net]
  The functor $\occEs{} : \occ \to \les$ is defined as follows. Let
  $O = (B,\pe{B},E, \delta_0, \delta_1,v_0)$ be an occurrence
  p-net. The corresponding {\evstr} is
  $\occEs{O} = (E, \vdash_{\unf{N}}, \#_{\unf{N}})$. For any morphism $f : O_1 \to O_2$ we let
  $\occEs{f} = \tr{f}$.
\end{defi}

The {\evstr} associated with a p-net is obtained from its
unfolding by forgetting the places and keeping the events and
their dependencies.  We first show that $\occEs{}$ is
well-defined on objects.

\begin{lem}[occurrence p-nets to locally connected {\evstr}]
  Let $O$
  be an occurrence
  p-net. Then $\occEs{O}$ is a locally connected {\es}.
\end{lem}

\begin{proof}
  Let $e \in \occEs{O}$ be an event. Since $O$ is an occurrence p-net,
  for all $b \in \pre{e}$ the pre-set $\pre{b}$ is connected. Moreover, if
  $\pre{b} \neq \emptyset$ then
  $\pre{b}$ is a disjunct for $e$. In fact, clearly, for each
  configuration $C \in \conf{\occEs{O}}$ such that $C \vdash e$ we
  have $C \cap \pre{b} \neq \emptyset$. Still $\pre{b}$ might not be a
  disjunct for the failure of minimality, i.e., for the existence of
  an event $e' \in \pre{b}$ such that $\pre{b} \setminus \{ e' \}$
  still intersects any configuration enabling $e$. However, it is easy
  to see that this would imply the existence of an event
  $e'' \in \pre{b}$ such that $e'' \struct_O e'$, violating the
  irredundancy assumption.

  It is also immediate to see that
  $\{ \pre{b} \mid b \in \pre{e}\ \land\ \pre{b} \neq \emptyset \}$ is
  a covering. Hence we conclude.  
\end{proof}

The fact that $\occEs{}$ is well-defined on morphisms is a consequence
of the lemma below.
For an occurrence p-net $O$, let us denote by $\prec_O$ immediate
causality on events, i.e., $e \prec_O e'$ if there is $b \in \pre{e'}$
such that $\pre{b} = \{e\}$. Note that causality
$\leq_O$ on events is the transitive closure of $\prec_O$.

\begin{lem}[properties of occurrence p-net morphisms]
  Let $O$, $O'$ be occurrence p-nets and $f : O \to O'$ a
  morphism. Then for all $C \in \conf{O}$ and $e_1, e_2 \in E$
  with $f(e_1)$, $f(e_2)$ defined
   \begin{itemize}
  \item if $\tr{f}(e_1) = \tr{f}(e_2)$ and  $e_1 \neq e_2$ then $e_1\,\#_{O}\, e_2$;
  \item if $\tr{f}(e_1)\,\#_{O'}\, \tr{f}(e_2)$ then $e_1\,\#_{O}\, e_2$;
  \item if $C \vdash_O e_1$ then $\tr{f}(C) \vdash_{O'} \tr{f}(e_1)$.
  \end{itemize}
  \end{lem}

\begin{proof}
  (1) Let $\tr{f}(e_1) = \tr{f}(e_2)$ and $e_1 \neq e_2$. Consider a
  causal chain of events in $O'$, starting from the initial marking
  and passing through non-persistent conditions only, namely
  consider $e_1', e_2', \ldots, e'_{n-1}, e_n' = \tr{f}(e_1) = \tr{f}(e_2)$ and
  conditions $b_1', \ldots, b_n' \in \npe{B'}$ such that
  $b_1 \in v_0'$ and for all $i \in \interval{n-1}$ it holds
  $b_i \in \post{e_i'} \cap \pre{e_{i+1}'}$; note that such a chain exists
  because $O'$ is
  t-restricted.
  It is easy to see that
  there must be corresponding causal chains $e_j^1, \ldots, e_j^n=e_j$
  and $b_j^1, \ldots, b_j^n \in \npe{B}$, for $j \in \{1,2\}$, such
  that $b_j^1 \in v_0$ and for all $i \in \interval{n-1}$ it
  holds $b_j^i \in \post{e_j^i} \cap \pre{e_j^{i+1}}$, which are
  mapped to the causal chain in $O'$, i.e., $\tr{f}(e_j^i) = e_i'$ and
  $b_i' \in \pl{f}(b_j^i)$ for all $i \in \interval{n}$,
  $j \in \{1,2\}$.

  Consider the least $i$ such that $e_1^i \neq e_2^i$. If $i=1$, i.e.,
  $e_1^1 \neq e_2^1$ then, since $b_1^1, b_2^1 \in v_0$,
  $b_1' \in \pl{f}(b_1^1)$ and $b_1' \in \pl{f}(b_2^1)$, recalling that
  $\pl{f}(v_0) = v_0'$ and $b_1'$, by safety of the initial marking,
  must have multiplicity at most $1$, we conclude that necessarily
  $b_1^1 = b_2^1$ and thus $e_1^1 \,\#_O\, e_2^1$, hence, by inheritance,
  $e_1=e_1^n\, \#_O\, e_2^n = e_2$, as desired.

  Otherwise, if $i>1$, i.e., $e_1^{i-1} = e_2^{i-1}$, then just
  observe that
  $b_1^{i-1}, b_2^{i-1} \in \post{e_1^{i-1}} = \post{e_2^{i-1}}$, and,
  by safety of $\post{e_j^{i-1}}$, as above, we conclude
  $e_1^i \,\#_O\, e_2^i$, and thus $e_1 \,\#_O\, e_2$, as desired.

  \bigskip  
  
  (2) Let $\tr{f}(e_1) \,\#_{O'}\, \tr{f}(e_2)$ in $O'$. We proceed by
  induction on the length of the derivation of the conflict. If such
  length is $0$, i.e., the conflict is direct, then there is
  $b' \in \pre{ \tr{f}(e_1)} \cap \pre{\tr{f}(e_2)} \cap \npe{B'}$.
  Thus, for $j \in \{1,2\}$, there are
  $b_j \in \pre{e_j} \cap \npe{B}$, such that $b' \in \pl{f}(b_j)$.
  If $b_1 =b_2$ then $e_1\,\#_O\, e_2$, and we conclude. Otherwise,
  observe that $b_1, b_2 \not \in v_0$, otherwise (again by safety of
  the initial marking) they could not have a common image. Therefore
  $\pre{b_1} = \{ e_1'\}$ and $\pre{b_2} = \{ e_2'\}$ for suitable
  $e_1, e_2 \in E$ (the pre-set is a singleton since $b_1$ and $b_2$
  are not persistent). Moreover
  $b' \in \post{\tr{f}(e_1')} = \post{\tr{f}(e_2')}$ hence it must be
  $\tr{f}(e_1') = \tr{f}(e_2')$. From point (1), $e_1' \,\#_O\, e_2'$ and
  thus $e_1 \,\#_O\, e_2$, by inheritance.

  If instead the conflict $\tr{f}(e_1) \,\#_{O'}\, \tr{f}(e_2)$ is
  inherited, we must have $b_1' \in \pre{\tr{f}(e_1)}$ and
  $b_2' \in \pre{\tr{f}(e_2)}$ such that $b_1' \#_{O'} b_2'$. In turn
  this means that $\pre{b_1'}, \pre{b_2'} \neq \emptyset$ and for all
  $e_1' \in \pre{b_1'}$, $e_2' \in \pre{b_2'}$, $e_1' \#_{O'} e_2'$.

  Consider $b_1 \in \pre{e_1}$, $b_2 \in \pre{e_2}$ such that
  $b_1' \in \pl{f}(b_1)$ and $b_2' \in \pl{f}(b_2)$. Necessarily,
  $b_1, b_2 \not\in v_0$, hence $\pre{b_1}, \pre{b_2} \neq
  \emptyset$. Moreover, for all $e_3 \in \pre{b_1}$,
  $\tr{f}(e_3) \in \pre{b_1'}$ and for all $e_4 \in \pre{b_2}$,
  $\tr{f}(e_4) \in \pre{b_2'}$. Hence, by the observation above,
  $\tr{f}(e_3) \,\#_{O'}\, \tr{f}(e_4)$, and the derivation of such
  conflict is shorter than that of
  $\tr{f}(e_1) \,\#_{O'}\, \tr{f}(e_2)$. Thus, by inductive
  hypothesis, $e_3 \,\#_O\, e_4$.  By definition of conflict, this
  implies that $b_1 \,\#_O\, b_2$ and thus $e_1 \,\#_O\, e_2$, as
  desired.

  (3) %
  By Definition~\ref{def:dependence} we have to show that for each $b' \in \pre{\tr{f}(e_1)}$ 
  either $b' \in u'_0$ or there is an $e' \in  \tr{f}(C)$ such that $b' \in \post{e'}$. Given
  $b' \in \pre{\tr{f}(e_1)}$, since $f$ is a morphism there is a $b \in \pre{e_1}$ such that 
  $b' \in \pl{f}(b)$. Since  $C \vdash_O e_1$, either $b \in u_0$, and in this case $b' \in 
   \pl{f}(u_0) = u'_0$, or there is an $e \in C$ such that $b \in \post{e}$. In the last case
   we have $b' \in \pl{f}(b) \subseteq \pl{f}(\post{e}) = \post{\tr{f}(e)}$, where $\tr{f}(e) \in \tr{f}(C)$, 
   as desired. 
\end{proof}

Conversely, we show how to freely generate an occurrence p-net from a
locally connected {\evstr}. Roughly, the idea is to insert suitable
conditions that induce exactly the dependencies (enabling and
conflict) of the original {\evstr}.

\begin{defi}[occurrence p-net for an {\evstr}]
\label{def:fromEsToPNet}
  Let $(E, \vdash, \#)$ be a
  locally connected
  {\evstr}. We define the occurrence p-net $\esOcc{E} = (B, \pe{B}, E, \gamma_0, \gamma_1,u_0)$  as follows.
  The set of places $B$ consists of
  \begin{itemize}
  \item non-persistent places $\langle X, Y \rangle$ with
    $X, Y \subseteq E$, $|X| \leq 1$ such that $e < e'$ for all
    $e \in X$, $e' \in Y$ and $e' \# e''$ for all $e', e'' \in Y$,
    $e' \neq e''$;

  \item persistent places $\langle X, Y \rangle$ with
    $X, Y \subseteq E$, $X$ disjunct of all $e \in Y$ and $X$
    $\conn{}$-connected.
  \end{itemize}
  Furthermore, for all $e \in E$ let $\gamma_0(e) = \{\langle X,Y \rangle \in B \mid e \in Y\}$, and 
  $\gamma_1(e) = \{\langle X,Y \rangle \in B \mid e \in X\}$.
 Finally, let the initial marking be $u_0 = \{ \langle \emptyset, Y \rangle \mid
  Y \subseteq E\ \land\ \langle \emptyset, Y \rangle \in B \}$.

\end{defi}

The intuition is the following.  For any possible set of events $Y$
pairwise in conflict that have a common cause $e$ we insert a
non-persistent place $b = \langle \{ e \}, Y \rangle$, produced by $e$
and consumed by the events in $Y$, inducing such dependencies. By the same clause, for any possible set of events $Y$
pairwise in conflict, we insert a
non-persistent place $b = \langle \emptyset, Y \rangle$ consumed by the events in $Y$.
Moreover, for any pair of sets of events $\langle X, Y \rangle$ such
that $X$ includes an event for each minimal enabling of each
$e \in Y$, we introduce a persistent condition $\langle X, Y \rangle$
that is generated by all events in $X$ and used by all events in $Y$.
In this way, whenever a minimal enabling set for some $e \in Y$ has
been executed, all the pre-set of $e$ is covered. Conversely, when all
the pre-set of $e$ is covered, since we generate conditions that
include %
at least one event for each minimal enabling of each $e \in Y$, at
least one minimal enabling has been completely executed.
The request that for $b = \langle X, Y \rangle \in \pe{B}$ the set
$X = \pre{b}$ is a disjunct connected by $\conn{}$ will ensure that
the pre-set of conditions is connected, as required by the definition
of occurrence p-net.  The fact that $X$ is a disjunct will
guarantee irredundancy.
Formally, the fact that the construction above produces a well-defined occurrence p-net will be a consequence of Lemma~\ref{le:dep-es-occ}.

We next observe that for all locally connected {\evstr}s, if we
build the corresponding occurrence p-net and then we take the
underlying {\evstr} we get an {\evstr} isomorphic to the
original one. First we prove a technical result.

\begin{lem}[dependencies in the occurrence p-net for an {\evstr}]
  \label{le:dep-es-occ}
  Let $E$ be a locally connected {\evstr} and
  $\esOcc{E} = (B, \pe{B}, E, \gamma_0, \gamma_1, v_0)$. Then
  \begin{enumerate}
    
  \item
    for all $e, e' \in E$, $e \# e'$ iff $e \#_{\esOcc{E}} e'$;

  \item
    for all $C \in \conf{E}$, $C \vdash e$ implies $C \vdash_{\esOcc{E}} e$;

  \item
    for all $C \in \conf{\esOcc{E}}$, $C \vdash_{\esOcc{E}} e$ implies
    $C \vdash e$.

  \end{enumerate}  
\end{lem}

\begin{proof}
  (1) \emph{Only if part.} Suppose that $e \# e'$ for $e, e' \in
  E$. Then, by Definition~\ref{def:fromEsToPNet}, $\esOcc{E}$ contains
  a non-persistent place
  $\langle \emptyset, \{e,e'\}\rangle \in \pre{e} \cap \pre{e'}$, thus
  $e \,\#_{\esOcc{E}}\, e'$.
  
  \emph{If part.} We prove, more generally, that for $A \subseteq E$,
  if $\#_{\esOcc{E}} A$ then there is no configuration
  $C \in \conf{E}$ such that $A \subseteq C$. Then for binary conflict
  the thesis follows from the fact that $E$ is saturated.
  
  We proceed by induction on the length of the derivation of the
  conflict. If the length is $0$, i.e., the conflict is direct,
  there are $e, e' \in A$ and
  $b \in \pre{e}\, \cap\, \pre{e'}\, \cap\, \npe{B}$. By
  Definition~\ref{def:fromEsToPNet} this means that
  $b = \langle \{e''\}, Y\rangle$ with $e, e' \in Y$ and therefore
  $e \# e'$ by construction and we are done.

  If instead the conflict $\#_{\esOcc{E}}\, A$ is inherited, we must
  have $e \in A$ such that
  $\#_{\esOcc{E}}\, (A \setminus \{ e\}) \cup \pre{e}$.  This means
  that if we consider any $A_e \subseteq E$ such that
  $A_e \cap \pre{b} \neq \emptyset$ for all $b \in \pre{e}$ then
  $\#_{\esOcc{E}}\, (A \setminus \{ e\}) \cup A_e$.  The derivations
  of the latter conflicts are shorter than that of
  $\#_{\esOcc{E}}\, A$, and thus by induction hypothesis we can infer
  that there is no configuration $C \in \conf{E}$ such that
  $(A \setminus \{ e\}) \cup A_e \subseteq C$.  
  Now, assume by absurd that $A \subseteq \hat{C}$ for a configuration
  $\hat{C}$. Since $e \in \hat{C}$, there is a configuration
  $\hat{C}_e \subset \hat{C}$ such that $\hat{C}_e \vdash e$.  By
  construction
  $\{ X \mid X\neq \emptyset\ \land\ b = \langle X, Y \rangle \in
  \pre{e} \} = \{ \pre{b} \mid b \in \pre{e}\ \land\ \pre{b} \neq
  \emptyset \}$ is a covering of $e$, therefore
  $\hat{C}_e \cap \pre{b} \neq \emptyset$ for all $b \in \pre{e}$.
  But clearly
  $(A \setminus \{ e \}) \cup \hat{C}_e \subseteq \hat{C}$,
  contradicting the inductive hypothesis.

  \bigskip
  
  (2) Let $C \in \conf{E}$ such that $C \vdash e$. In order to
  conclude that $C \vdash_{\esOcc{E}} e$ we have to prove that for all
  $b \in \pre{e}$, either $b \in v_0$ or there exists $e' \in C$ such
  that $b \in \post{e'}$. We distinguish two cases.

  If $b$ is not persistent, then $b = \langle \{e'\}, Y \rangle$ with
  $e \in Y$, hence $e' < e$. Then  $b \in \post{e'}$.

  If $b$ is persistent, then $b = \langle X, Y \rangle$ with $e \in Y$ and, by
  definition of $\esOcc{E}$,
  $X \cap C \neq \emptyset$. Let
  $e' \in X \cap C$. Then $e' \in C$ and $b \in \post{e'}$.

  \bigskip

  (3) We prove that for all finite $C \in \conf{\esOcc{E}}$ it holds
  that $C \in \conf{E}$ and, for any $e \in E$, if
  $C \vdash_{\esOcc{E}} e$ then $C \vdash e$.

  The proof proceeds by induction on the cardinality $|C|$. The base
  case $|C|=0$ is immediate. If $|C|>0$, we know by
  Lemma~\ref{le:conf-exec} that $C = \{ e_1, \ldots, e_n \}$ and there
  exists a firing sequence
  $v_0 \fire{e_1} v_1 \fire{e_2} \ldots \fire{e_n} v_n$. Since
  $C' = \{ e_1, \ldots, e_{n-1} \} \in \conf{\esOcc{E}}$ and
  $C' \vdash_{\esOcc{E}} e_n$, by inductive hypothesis
  $C' \in \conf{E}$ and $C' \vdash e_n$. From this and point (1), we
  deduce that $C = C' \cup \{ e\} \in \conf{E}$.

  For the second part, let $e \in E$ be such that $C \vdash_{\esOcc{E}} e$.
  By the first part, $C \in \conf{E}$ and by definition of enabling in
  occurrence p-nets, for all $b \in \pre{e}$, either $b \in v_0$ or
  there exists $e' \in C$ such that $b \in \post{e'}$.
  This means that $C \cap \pre{b} \neq \emptyset$ for all
  $b \in \pre{e}$, $\pre{b} \neq \emptyset$ and this is a covering of
  $e$. Hence by definition of covering, $C \vdash e$.
\end{proof}

\begin{cor}[$\esOcc{E}$ is well-defined]
  Let $E$ be a locally connected {\evstr}. Then $\esOcc{E}$ is a well-defined occurrence p-net.
\end{cor}

\begin{proof}
  The fact that $\esOcc{E}$ satisfies the structural properties (a)
  and (c) of Definition~\ref{def:occ-p-net} is immediate by
  construction. Property (b), i.e., the fact that each transition
  $e \in T$ admits a securing sequence follows by the analogous
  property of {\evstr}, recalling that by Lemma~\ref{le:dep-es-occ},
  enabling and conflict coincide in $E$ and $\esOcc{E}$. Similarly,
  the fact that for each condition $b \in B$, the set $\pre{b}$ is
  connected is true by construction after Lemma~\ref{le:dep-es-occ},
  showing that conflict, and thus the notion of connectedness,
  coincide in $E$ and $\esOcc{E}$.
  Finally, we observe that $\esOcc{E}$ is
  well-formed. T-restrictedness holds by construction. Irredundancy
  follows from the fact that for each
  $b = \langle X, Y \rangle \in \pe{B}$, the pre-set $\pre{b} = X$ is
  a disjunct. Hence given $e \in X = \pre{b}$ there cannot be an
  additional path $e \struct_{\esOcc{E}}^n$ with $n \geq 2$ into
  $b$. In fact, if this were the case, since by
  Lemma~\ref{le:dep-es-occ}, enabling coincides in $E$ and
  $\esOcc{E}$, event $e$ could be omitted, i.e.,
  $X' = X \setminus \{e\}$ would still intersect any configuration
  enabling $e$, contradicting the minimality of $X$.
\end{proof}

\begin{cor}[unit]
  \label{co:es-unit}
  Let $E$ be a locally connected {\evstr}. Then
  $\eta_E : E \to \occEs{\esOcc{E}}$ defined as the identity on events
  is an isomorphism.
\end{cor}

\begin{proof}
  The function $\eta_E$ is obviously a bijection. The fact that it is
  an isomorphism of {\evstr} follows immediately from
  Lemma~\ref{le:dep-es-occ}.
\end{proof}

In order to conclude we need to show that the construction of the
occurrence p-net associated with an {\evstr} $E$ and the
isomorphism $\eta_E : E \to \occEs{\esOcc{E}}$ are free over $E$.  The
next lemma states some properties of occurrence p-net morphisms that guide
the proof.

\begin{lem}
  \label{le:opnet-morph-from-events}
  Let $h : O \to O'$ be an occurrence p-net morphism. Then for
  $b \in B$ and $b' \in B'$ such
  that $b' \in \pl{h}(b)$ it holds
  \begin{enumerate}
  
  \item if $b \in \npe{B}$ then (a) $\tr{h}(\pre{b}) = \pre{b'}$ and
    $\post{b} = \{ e_1 \in E \mid e_1 \in \tr{h}^{-1}(\post{b})\
    \land\ \pre{b} \leq e_1 \}$, where $\pre{b} \leq e_1$ means
    $e \leq e_1$ when $\pre{b} = \{e\}$ and it is a vacuous
    requirement, otherwise;

  \item if $b \in \pe{B}$ then (a)
    $\pre{b} \subseteq \tr{h}^{-1}(\pre{b'})$ and (b)
    $\post{b} \subseteq \tr{h}^{-1}(\post{b'})$.
  \end{enumerate}
\end{lem}

\begin{proof}
  Point (1) is analogous to the one in the ordinary case. Point (2)
  easily follows from the definition of morphism.
\end{proof}  

\begin{thm}[coreflection between $\esOcc{}$ and $\occEs{}$]\label{th:OccEs}
  The construction $\occEs{}$ extends to a functor that is left adjoint
  to $\esOcc{}$, and they establish a coreflection.
\end{thm}

\begin{proof}

  Let $E$ be a locally connected {\evstr}. We show the freeness of
  $\eta_E : E \to \occEs{\esOcc{E}}$ as defined in
  Corollary~\ref{co:es-unit}.  We have to show that for any occurrence
  p-net
  $O' = \langle B', \pe{B'}, E', \gamma_0', \gamma_1', v_0' \rangle$
  and for any {\evstr} morphism $f :E \to \occEs{O'}$ there exists a
  unique morphism $h: \esOcc{E} \to O'$ such that the following
  diagram commutes
  \[
  \xymatrix{
    {E}  \ar[r]^(.40){\eta_{E}} \ar[dr]_f & {\occEs{\esOcc{E}}} \ar@{.>}[d]^{\occEs{h}} \\
    & {\occEs{O'}} 
    }
  \]

  The transition component of $h$ is determined as $\tr{h} = f$. It
  can be extended to a p-net morphism by defining the place component as
  follows:

  \begin{enumerate}
  \item on non-persistent places $b = \langle X, Y \rangle$,
    define
    \begin{center}
      $\pl{h}(b) = \{ b' \in \npe{B'} \mid \pre{b'} = \tr{h}(X)\
      \land\ Y = \{ e_1 \in E \mid e_1 \in
      \tr{h}^{-1}(\post{b'})\land\ X < e_1\}\}$.
    \end{center}

  \item on persistent places $b = \langle X, Y \rangle$, say that $b$ is a potential pre-image of $b' \in \npe{B'}$ if $X \subseteq \tr{h}^{-1}(\pre{b'})$ and $Y \subseteq \tr{h}^{-1}(\post{b'})$. It is maximal if for any potential pre-image $b_1 = \langle X_1, Y_1 \rangle$ of $b'$, with $X_1 \cap X \neq \emptyset$ it holds $Y_1 \subseteq Y$. Then we define
    \begin{center}
      $      
        \pl{h}(b) = \{ b' \in \npe{B'} \mid \mbox{$b$ maximal potential pre-image of $b'$}\}$.
    \end{center}
  \end{enumerate}

  It is not difficult to prove that this is a morphism. Uniqueness
  follows from the fact that, for non-persistent places (point (1)) we
  defined the morphism in the only possible way, according to
  Lemma~\ref{le:opnet-morph-from-events}(1). For persistent places, each $b' \in \pe{B'}$ must be in the image of exactly one potential pre-image. In fact, according to Lemma~\ref{le:opnet-morph-from-events}(2), it must be in the image of at least one potential pre-image and, by injectivity on pre- and post-sets of transitions, it can be in the image of at most one of the potential pre-images. Using preservation of pre- and post-sets, we then conclude that $\pl{h}(s)$ must be defined as we did in point (2).
\end{proof}

As a side remark, note that the {\evstr} induced by $\preUnf{N}$
(before quotienting) is a prime {\evstr}. If endowed with the
equivalence $\sim_N$ restricted to events, it corresponds to the notion
of prime {\evstr} with equivalence~\cite{VW:SPC,BCG:DESF}. More
precisely, if we denote by $\#_N^\mu$ the immediate conflict relation,
then the prime {\evstr} obtained from the pre-unfolding, by forgetting
the conditions and the direct conflicts between equivalent events,
i.e.,  $\langle (E,\leq_{\unf{N}}, \#), \sim \rangle$, where
$\# = \#_{\unf{N}} \setminus (\#_{\unf{N}}^\mu \cap \sim_N)$
and $\sim = \sim_N \cap (E \times E)$, is a prime {\evstr} in the
sense of~\cite{BCG:DESF}. Its events arise, as explained
in~\cite{BCG:DESF}, as the irreducible elements of the domain
associated with $N$.

\section{Conclusions}
\label{sec:conc}

Persistence is the continuance of an effect after its causes ceased to exist.
In this paper we have studied the effect of adding persistence to Petri nets at the level of {\evstr} semantics.
Interestingly, we have extended Winskel's chain of coreflection from the category of p-nets to the newly defined category of locally connected {\evstr}s, which is a full subcategory of the category of general {\evstr}s.
Since the category of connected {\evstr}s is included in the one of locally connected {\evstr}s, the coreflection can serve to explain in basic terms the phenomenon of fusion arising in the context of graph grammars and that induces (connected) disjunctive causes. On the one hand, this confirms our intuition that Petri nets and their natural extensions keep capturing all phenomena of concurrency within easy-to-understand operational models. On the other hand, our results show that while non-prime {\evstr}s were actually underestimated in the literature, at least in some cases they are natural, expressive and equipped with an interesting theory even at the operational level.

The result has been proved for the class of well-formed
persistent nets, where redundant paths to persistent places are
forbidden. Despite the fact that this is a natural restriction,
preliminary investigations suggest that the result could be extended,
at the price of some technical complications, to the more general
setting.

\section*{Acknowledgment}
\noindent
The authors wish to thank the anonymous reviewers for their careful
reading and their insightful comments on the submitted version on the
paper.

\bibliographystyle{plain}
\bibliography{PersistentNets}

\begin{thebibliography}{10}

\bibitem{DBLP:conf/birthday/Abramsky08}
Samson Abramsky.
\newblock Petri nets, discrete physics, and distributed quantum computation.
\newblock In Pierpaolo Degano, Rocco~De Nicola, and Jos{\'{e}} Meseguer,
  editors, {\em Concurrency, Graphs and Models, Essays Dedicated to Ugo
  Montanari on the Occasion of His 65th Birthday}, volume 5065 of {\em Lecture
  Notes in Computer Science}, pages 527--543. Springer, 2008.

\bibitem{DBLP:journals/mscs/AspertiB09}
Andrea Asperti and Nadia Busi.
\newblock Mobile {P}etri nets.
\newblock {\em Mathematical Structures in Computer Science}, 19(6):1265--1278,
  2009.

\bibitem{BCG:DESF}
Paolo Baldan, Andrea Corradini, and Fabio Gadducci.
\newblock Domains and event structures for fusions.
\newblock In {\em LICS 2017}, pages 1--12. {IEEE} Computer Society, 2017.
\newblock Extended version available as
  \href{https://arxiv.org/abs/1701.02394}{arXiv:1701.02394}.

\bibitem{BCM:CNAED}
Paolo Baldan, Andrea Corradini, and Ugo Montanari.
\newblock Contextual {P}etri nets, asymmetric event structures and processes.
\newblock {\em Information and Computation}, 171(1):1--49, 2001.

\bibitem{DBLP:journals/fuin/BonchiBCG09}
Filippo Bonchi, Antonio Brogi, Sara Corfini, and Fabio Gadducci.
\newblock A net-based approach to web services publication and replaceability.
\newblock {\em Fundamenta Informaticae}, 94(3-4):305--330, 2009.

\bibitem{DBLP:journals/corr/abs-1710-04570}
Roberto Bruni, Hern{\'{a}}n~C. Melgratti, and Ugo Montanari.
\newblock Concurrency and probability: Removing confusion, compositionally.
\newblock In {\em LICS 2018}, pages 195--204. ACM, 2018.

\bibitem{DBLP:conf/apn/ChristensenH93}
S{\o}ren Christensen and Niels~Damgaard Hansen.
\newblock Coloured {P}etri nets extended with place capacities, test arcs and
  inhibitor arcs.
\newblock In Marco~Ajmone Marsan, editor, {\em ICATPN 1993}, volume 691 of {\em
  Lecture Notes in Computer Science}, pages 186--205. Springer, 1993.

\bibitem{DBLP:conf/ccs/CrazzolaraW01}
Federico Crazzolara and Glynn Winskel.
\newblock Events in security protocols.
\newblock In Michael~K. Reiter and Pierangela Samarati, editors, {\em CCS
  2001}, pages 96--105. {ACM}, 2001.

\bibitem{CrazzolaraW05}
Federico Crazzolara and Glynn Winskel.
\newblock Petri nets with persistence.
\newblock In Nadia Busi, Roberto Gorrieri, and Fabio Martinelli, editors, {\em
  WISP 2004}, volume 121 of {\em Electronic Notes in Theoretical Computer
  Science}, pages 143--155. Elsevier, 2005.

\bibitem{VW:SPC}
Marc de~Visme and Glynn Winskel.
\newblock Strategies with parallel causes.
\newblock In Valentin Goranko and Mads Dam, editors, {\em {CSL} 2017},
  volume~82 of {\em LIPIcs}, pages 41:1--41:21. Schloss Dagstuhl -
  Leibniz-Zentrum f\"ur Informatik, 2017.

\bibitem{DBLP:conf/performance/DuganTGN84}
Joanne~Bechta Dugan, Kishor~S. Trivedi, Robert Geist, and Victor~F. Nicola.
\newblock Extended stochastic {P}etri nets: Applications and analysis.
\newblock In Erol Gelenbe, editor, {\em Performance 1984}, pages 507--519.
  North-Holland, 1985.

\bibitem{DBLP:conf/apn/EmzivatDLR16}
Yrvann Emzivat, Beno{\^{\i}}t Delahaye, Didier Lime, and Olivier~H. Roux.
\newblock Probabilistic time {P}etri nets.
\newblock In Fabrice Kordon and Daniel Moldt, editors, {\em ICATPN 2016},
  volume 9698 of {\em Lecture Notes in Computer Science}, pages 261--280.
  Springer, 2016.

\bibitem{DBLP:books/daglib/0007558}
Claude Girault and R{\"{u}}diger Valk.
\newblock {\em Petri Nets for Systems Engineering - A Guide to Modeling,
  Verification, and Applications}.
\newblock Springer, 2003.

\bibitem{DBLP:series/eatcs/Gorrieri17}
Roberto Gorrieri.
\newblock {\em Process Algebras for Petri nets - The Alphabetization of
  Distributed Systems}.
\newblock Springer, 2017.

\bibitem{DBLP:conf/ac/Jensen86a}
Kurt Jensen.
\newblock Coloured {P}etri nets.
\newblock In Wilfried Brauer, Wolfgang Reisig, and Grzegorz Rozenberg, editors,
  {\em Petri Nets: Central Models and Their Properties, Advances in Petri Nets
  1986, Part I}, volume 254 of {\em Lecture Notes in Computer Science}, pages
  248--299. Springer, 1987.

\bibitem{DBLP:books/daglib/0023756}
Kurt Jensen and Lars~Michael Kristensen.
\newblock {\em Coloured Petri Nets - Modelling and Validation of Concurrent
  Systems}.
\newblock Springer, 2009.

\bibitem{DBLP:journals/sosym/Koch15}
Ina Koch.
\newblock Petri nets in systems biology.
\newblock {\em Software and System Modeling}, 14(2):703--710, 2015.

\bibitem{DBLP:journals/it/KochRS14}
Ina Koch, Wolfgang Reisig, and Falk Schreiber.
\newblock Petri nets in the biosciences.
\newblock {\em it - Information Technology}, 56(2):43--45, 2014.

\bibitem{DBLP:journals/jacm/LandweberR78}
Lawrence~H. Landweber and Edward~L. Robertson.
\newblock Properties of conflict-free and persistent {P}etri nets.
\newblock {\em Journal of the {ACM}}, 25(3):352--364, 1978.

\bibitem{DBLP:journals/tocs/MarsanCB84}
Marco~Ajmone Marsan, Gianni Conte, and Gianfranco Balbo.
\newblock A class of generalized stochastic {Petri} nets for the performance
  evaluation of multiprocessor systems.
\newblock {\em ACM Transactions on Computer Systems}, 2(2):93--122, 1984.

\bibitem{merlin74}
Philip~M. Merlin.
\newblock {\em A study of the recoverability of computing systems}.
\newblock PhD thesis, University of California, Irvine, CA, 1974.

\bibitem{DBLP:conf/concur/MeseguerMS92}
Jos{\'{e}} Meseguer, Ugo Montanari, and Vladimiro Sassone.
\newblock On the semantics of {P}etri nets.
\newblock In Rance Cleaveland, editor, {\em CONCUR 1992}, volume 630 of {\em
  Lecture Notes in Computer Science}, pages 286--301. Springer, 1992.

\bibitem{DBLP:journals/mscs/MeseguerMS97}
Jos{\'{e}} Meseguer, Ugo Montanari, and Vladimiro Sassone.
\newblock On the semantics of place/transition {P}etri nets.
\newblock {\em Mathematical Structures in Computer Science}, 7(4):359--397,
  1997.

\bibitem{Molloy:1985:DTS:4101.4110}
Michael~K. Molloy.
\newblock Discrete time stochastic {Petri} nets.
\newblock {\em IEEE Transaction on Software Engineering}, 11(4):417--423, April
  1985.

\bibitem{DBLP:journals/acta/MontanariR95}
Ugo Montanari and Francesca Rossi.
\newblock Contextual nets.
\newblock {\em Acta Informatica}, 32(6):545--596, 1995.

\bibitem{murata89}
Tadao Murata.
\newblock Petri nets: Properties, analysis and applications.
\newblock {\em Proceedings of the IEEE}, 77(4):541--580, 1989.

\bibitem{Petri62}
Carl~Adam Petri.
\newblock {\em Kommunikation mit Automaten}.
\newblock PhD thesis, Universit\"at Hamburg, 1962.

\bibitem{ramchandani74}
Chander Ramchandani.
\newblock {\em Analysis of asynchronous concurrent systems by timed Petri
  nets}.
\newblock PhD thesis, Massachusetts Institute of Technology, Cambridge, MA,
  1973.

\bibitem{DBLP:books/daglib/0032298}
Wolfgang Reisig.
\newblock {\em Understanding Petri Nets - Modeling Techniques, Analysis
  Methods, Case Studies}.
\newblock Springer, 2013.

\bibitem{DBLP:conf/apn/Valk98}
R{\"{u}}diger Valk.
\newblock Petri nets as token objects: An introduction to elementary object
  nets.
\newblock In J{\"{o}}rg Desel and Manuel~Silva Su{\'{a}}rez, editors, {\em
  ICATPN 1998}, volume 1420 of {\em Lecture Notes in Computer Science}, pages
  1--25. Springer, 1998.

\bibitem{DBLP:books/daglib/0032051}
Wil M.~P. van~der Aalst and Christian Stahl.
\newblock {\em Modeling Business Processes - {A} Petri Net-Oriented Approach}.
\newblock {MIT} Press, 2011.

\bibitem{VSY:UFP}
Walter Vogler, Alexei~L. Semenov, and Alexandre Yakovlev.
\newblock Unfolding and finite prefix for nets with read arcs.
\newblock In Davide Sangiorgi and Robert de~Simone, editors, {\em CONCUR 1998},
  volume 1466 of {\em Lecture Notes in Computer Science}, pages 501--516.
  Springer, 1998.

\bibitem{Win:ES}
Glynn Winskel.
\newblock Event {S}tructures.
\newblock In Wilfried Brauer, Wolfgang Reisig, and Grzegorz Rozenberg, editors,
  {\em Petri Nets: Central Models and Their Properties, Advances in Petri Nets
  1986, Part II}, volume 255 of {\em Lecture Notes in Computer Science}, pages
  325--392. Springer, 1987.

\end{thebibliography}

\end{document}